\documentclass[11pt]{article}

\usepackage{./base}

\newcommand{\email}[1]{\href{mailto:#1}{\texttt{#1}}}
\newcommand*\samethanks[1][\value{footnote}]{\footnotemark[#1]}
\title{Robust Radical Sylvester-Gallai Theorem for Quadratics}

\author{
  Abhibhav Garg
  \thanks{
    University of Waterloo, \ \ 
    \email{a65garg@uwaterloo.ca}, \ \ \email{rafael@uwaterloo.ca}
  }
  \and
  Rafael Oliveira 
  \samethanks  
  \and
  Akash Kumar Sengupta
  \thanks{
    Columbia University,
    \email{akashs@math.columbia.edu}
  }
}

\begin{document}

\maketitle

\begin{abstract}
We prove a robust generalization of a Sylvester-Gallai type theorem for quadratic polynomials, generalizing the result in \cite{S19}.
More precisely, given a parameter $0 < \delta \leq 1$ and a finite collection $\cF$ of irreducible and pairwise independent polynomials of degree at most 2, we say that $\cF$ is a $(\delta, 2)$-radical Sylvester-Gallai configuration if for any polynomial $F_i \in \cF$, there exist $\delta(|\cF| -1)$ polynomials $F_j$ such that $|\radideal{F_i, F_j} \cap \cF| \geq 3$, that is, the radical of $F_i, F_j$ contains a third polynomial in the set.

In this work, we prove that any $(\delta, 2)$-radical Sylvester-Gallai configuration $\cF$ must be of low dimension: that is 
$$\dim \Cspan{\cF} = \poly{1/\delta}.$$
\end{abstract}

\section{Introduction}\label{sec:intro}


Suppose $v_{1}, \dots, v_{m} \in \bR^{n}$ is a set of distinct points, such that the line joining any two points in the set contains a third point.
In 1893, Sylvester asked if such configurations of points are necessarily colinear \cite{sylvester1893mathematical}.
Independently, this same question was asked by Erdös in 1943 \cite{erdos1943problems}.
This was independently proved by \cite{melchior1940uber, gallai1944solution}, and this result is known as the Sylvester-Gallai theorem.
A set of points satisfying the above property is called a Sylvester-Gallai (SG) configuration.

Sylvester-Gallai theorems depend on the base field.
For instance, it is well known that any nonsingular planar cubic curve over $\bC$ has nine inflection points, and that any line passing through two such points passes through a third \cite{dickson1914points}.
These nine points are not collinear, and therefore form a counterexample to the Sylvester-Gallai theorem when the underlying field is changed from $\bR$ to $\bC$.
In 1966, Serre asked if there are configuration of points in $\bC^{n}$ that satisfy the Sylvester-Gallai that are not coplanar \cite{serre1966advanced}. In 1986, Kelly noted that no such configurations can exist, namely, that points in $\bC^{n}$ that satisfy the Sylvester-Gallai property must be coplanar \cite{kelly1986resolution}. Kelly showed that the planarity of Sylvester-Gallai configurations is a simple consequence of Hirzebruch's work on line arrangements which relies on deep results from algebraic geometry \cite{H83}.

Over finite fields, Sylvester-Gallai configurations do not have bounded dimension.
For example, if we are working over the field $\bF_p$ (with $p > 2$) and our vector space is $\bF_p^n$, then the set of points is $\bF_{p}^{n}$ is a SG configuration of dimension $n$, which is not constant.
In general, any subgroup of $\bF_{p}^{n}$ will form a Sylvester-Gallai configuration.
Some bounds on the dimension of configurations in this setting can be found in \cite{dvir2012incidence}.
In this work, we only focus on fields of characteristic zero, and to make the presentation easier we restrict our attention to $\bC$.

Several variations and generalizations of the Sylvester-Gallai problem defined above have been studied in combinatorial geometry.
The underlying theme in all these types of questions is the following: 
\begin{center}
Are Sylvester-Gallai type configurations always low-dimensional?
\end{center}
In characteristic zero, the answer has always turned out to be yes.
For a thorough survey of the earlier works on SG-type theorems, we refer the reader to  \cite{borwein1990survey} and results therein.

While the above results are mathematically beautiful and interesting in their own right, it is also interesting and useful in areas such as computer science and coding theory to consider higher-dimensional analogs as well as robust analogs of SG-type theorems.

\paragraph{Higher-dimensional analogs of SG configurations}
In \cite{hansen1965generalization}, a higher dimensional version of the SG theorem was proved, with lines replaced by flats.
This variant has applications in the study of algebraic circuits, and in particular in Polynomial Identity Testing (PIT) \cite{kayal2009blackbox, saxena2013sylvester}, a central problem in algebraic complexity theory.
The works \cite{kayal2009blackbox, saxena2013sylvester} use the higher dimensional Sylvester-Gallai theorems to bound the ``rank" of certain types of depth three circuits.\footnote{Algebraic circuits which compute polynomials that can be written as a sum of products of linear forms.} 
In simple terms, if the linear forms of a circuit satisfy the high dimensional SG condition, then in essence the polynomial being computed must depend on a constant number of variables, in which case it is easy to check whether the circuit is computing a non-zero polynomial.

\paragraph{Robust analogs of SG configurations and applications:}
Robust generalizations of the Sylvester-Gallai theorem have found applications in coding theory and in complexity theory.

In this variant, for every point $v_{i}$ there are at least $\delta (m-1)$ points $u_{1}, \dots, u_{k}$ in the configuration such that $v_{i}$ and $u_{j}$ span a third point, for $1 \leq j \leq k$.
The usual Sylvester-Gallai theorem is the case when $\delta = 1$.
Such configurations were first studied by Szemerédi and Trotter \cite{szemeredi1983extremal}, who proved that if $\delta$ is bigger than an absolute constant close to $1$, then the configuration has constant dimension.

In \cite{BDWY11}, the authors prove that such a configuration has dimension $\bigO{1 / \delta^{2}}$, for any $0 < \delta \leq 1$.
This robust version also allows them to prove robust versions of the higher dimensional variants mentioned above.
They also define the notion of a $LCC$-configuration, which is an extension of the Sylvester-Gallai configuration where points are allowed to occur with multiplicity.
In \cite{DSW14}, the authors improve the bound on the dimension of robust Sylvester-Gallai configurations to $\bigO{1 / \delta}$.

In coding theory, these robust configurations naturally appear in the study of locally decodable codes and locally correctable codes \cite{BDWY11}.
These results, as well as similar results, are surveyed in \cite{dvir2012incidence}.
Robust SG configurations also have applications in the study of algebraic circuits, in particular in reconstruction of algebraic circuits \cite{sinha2016reconstruction}.

\paragraph{Higher degree generalizations of Sylvester-Gallai configurations:}
Also motivated by the PIT problem, Gupta in \cite{gupta2014algebraic} introduced higher degree generalizations of SG configurations, and asked if they are also ``low dimensional.''
In his paper, Gupta outlines a series of SG-type conjectures, and gives a deterministic polynomial-time blackbox PIT algorithm for a special class of algebraic circuits\footnote{These are circuits computed by a sum of constantly many products of constant degree polynomials.} assuming that these conjectures hold.

The first challenge in Gupta's series of conjectures on SG type theorems is the following:  
\begin{conjecture}[Conjecture~29, \cite{gupta2014algebraic}]
    \label{conjecture:firstdsg}
    Let $Q_{1}, \dots, Q_{m} \in \bC\bs{x_{1}, \dots, x_{n}}$ be irreducible, homogeneous, and of degree at most $d$ such that for every pair $Q_{i}, Q_{j}$ there is a $k$ such that $Q_{k} \in \radideal{Q_{i}, Q_{j}}$.
    Then the transcendence degree of $Q_{1}, \dots, Q_{m}$ is $\bigO{1}$ (where the constant depends on the degree $d$).
\end{conjecture}
The case $d = 2$ of the above conjecture was proved in \cite{S19}.
We henceforth refer to the original Sylvester-Gallai theorem (the case $d=1$) and its variants as the ``linear case".
As in the linear case of SG type problems, it is natural to consider the robust version of the above lemma as a next step towards the conjectures of Gupta that give an algorithm for a special case of PIT.
We resolve the robust version of the above in the case when $d=2$.

\subsection{Main result}

In this section, we formally state our main result: robust quadratic radical Sylvester-Gallai configurations must lie in a constant dimensional vector space.\footnote{Our results hold for any algebraically closed field of characteristic zero. However, for simplicity of exposition, we only state our results over $\bC$.}
In particular, this result implies that the forms must be contained in a small algebra, and also that they have constant transcendence degree.
Another important result is a structural result for ideals generated by two quadratic forms.

\subsubsection{Robust radical Sylvester-Gallai theorem}

We first formally define robust quadratic radical Sylvester-Gallai configurations.
As is customary in the literature, we will use \emph{form} to denote homogeneous polynomials. For a polynomial ring $S=\bC[x_1,\ldots, x_n]$, we let $S_d$ denote the vector space of polynomials of degree $d$ in $S$, and $(f_1,\cdots,f_r)$ denotes the ideal generated by polynomials $f_1,\cdots, f_r$.
We also use $\radideal{f_{1}, \dots, f_{r}}$ to denote the radical of this ideal, that is, the set of polynomials $g$ such that $g^{k} \in \ideal{f_{1}, \dots, f_{r}}$ for some $k$.

\begin{restatable}[$\drsg{\delta}{2}$ configurations]{definition}{deltasgconfig}
\label{def:deltasg2config}
Let $0 < \delta \leq 1$ and $\cF := \{F_1, \ldots, F_m\}$ be a set of irreducible forms in $\bC[x_1,\ldots, x_n]$. 
We say that $\cF$ is a $\drsg{\delta}{2}$ configuration if the following conditions hold: 
\begin{enumerate}
    \item $\cF \subset S_1 \cup S_2$ \hfill (only linear and quadratic forms)
    \item for any $i \neq j$, we have that $F_i \not\in (F_j)$ \hfill  (forms are ``pairwise independent'')
    \item for any $i \in [m]$, there are $\delta (m-1)$ indices $j \in [m] \setminus \bc{i}$ such that $\abs{\Rad(F_i, F_j) \cap \cF} \geq 3$.
\end{enumerate}
\end{restatable}

\noindent We are now ready to formally state the main contribution of our paper.
We begin with our main theorem, that robust quadratic radical SG configurations must have small linear span.

\vspace{5pt}

\begin{restatable}[$\drsg{\delta}{2}$ theorem]{theorem}{deltasgstatement}
\label{thm:deltasg2}
If $\cF$ is a $\drsg{\delta}{2}$ configuration, then 
\begin{equation*}
\dim(\Cspan{\cF}) = O(1/\delta^{54}).
\end{equation*}
\end{restatable}

To prove the theorem above, we first notice that the theorem would imply that the forms in the configuration are contained in a subalgebra of the polynomial ring of small dimension.
With this observation at hand, we provide a principled approach to construct small dimensional subalgebras of the polynomial ring which control the configuration (in the sense that all forms in the configuration will become a ``univariate form'' with coefficients from our subalgebra). 

The main property of these algebras is that they allow us to translate non-linear SG dependencies (the radical dependencies) into linear SG dependencies, and therefore we can reduce our non-linear problem to the linear version of the SG problem.

The main principle guiding the construction of our subalgebras is that we would like these subalgebras to look ``as free as possible'' without increasing the dimension of the algebra by much.
The amount of ``freeness'' that we need is captured by the robust algebras defined in \cref{sec:clean-ideals}, where we also elaborate on how these algebras behave with SG configurations (where we need the notion of clean algebras).
For more intuition about these algebras and about our strategy to prove our main theorem, we refer the reader to \cref{sec:proof-overview}.

\subsubsection{Results on structure of ideals generated by two quadratics}

A key step in our strategy to prove that a $\drsg{\delta}{2}$ configuration is low dimensional (as has also been the first step in the works of \cite{S19, PS20a}) is to understand the structure of ideals generated by two quadratic forms. 

The general principle at play here is that if the ideal generated by two quadratic forms is neither radical nor prime, then there must be a low-rank quadratic in their span. 
In \cite{S19, PS20a}, the authors proved similar structural results to determine when a product of quadratic forms is contained in an ideal generated by two quadratic forms. 
In \cref{theorem: radical structure}, we use a different approach to completely characterize when the ideal generated by two quadratic forms is radical or prime, and as corollaries we obtain the structural results in \cite{S19, PS20a} (see \cref{sec:algebraic-geom}). We use a commutative-algebraic approach to develop a further understanding of the radical of ideals generated by two irreducible quadratics. Indeed, using the standard tools of primary decomposition and Hilbert-Samuel multiplicity we obtain a classification of the possible minimal primes of an ideal generated by two quadratic forms. Consequently we obtain a characterization for such an ideal to be prime or radical. This approach can also be generalized to ideals generated by cubic forms, as was done in \cite{OS21}.

\begin{restatable}[Radical Structure Theorem]{proposition}{radicalstructurethm}
\label{theorem: radical structure}
Let $\bK$ be an algebraically closed field of characteristic zero and $Q_1,Q_2\in S=\bK[x_1,\cdots,x_n]$ be two forms of degree $2$. Then one of the following holds:
\begin{enumerate}
    \item The ideal $(Q_1,Q_2)$ is prime.
    \item The ideal $(Q_1,Q_2)$ is radical, but not prime. Furthermore, one of the following cases occur:
    \begin{itemize}
        \item [(a)] There exist two linearly independent linear forms $x,y\in S_1$ such that $xy\in \mathrm{span}(Q_1,Q_2)$.
        \item[(b)] There exists a minimal prime $\mathfrak{p}$ of $(Q_1,Q_2)$, such that $\mathfrak{p}=(x,y)$ for some linearly independent forms $x,y\in S_1$
    \end{itemize}
    \item The ideal $(Q_1,Q_2)$ is not radical and one of the following cases occur:
    \begin{itemize}
        \item [(a)] $Q_1,Q_2$ have a common factor and $Q_1=xy$,  $Q_2=x(\alpha x+\beta y)$ for some linear forms $x,y$ and $\alpha,\beta \in k$. In this case, we have $x^2\in \mathrm{span}(Q_1,Q_2)$.
        \item [(b)] $Q_1,Q_2$ do not have a common factor. There exists a minimal prime $\mathfrak{p}$ of $(Q_1,Q_2)$ such that $\mathfrak{p}=(x,Q)$, where $x\in S_1$, $Q\in S_2$ and $Q$ is irreducible modulo $x$, and we also have $x^2\in \mathrm{span}(Q_1,Q_2)$.
        \item [(c)] $Q_1,Q_2$ do not have a common factor and there exists a minimal prime $\mathfrak{p}$ of $(Q_1,Q_2)$, such that $\mathfrak{p}=(x,y)$ for some linearly independent forms $x,y\in S_1$, and the $(x,y)$-primary ideal $\mathfrak{q}$ has multiplicity $e(S/\mathfrak{q})\geq 2$.
    \end{itemize}
\end{enumerate}
\end{restatable}

The proposition above is not new, and proofs of some of the statements can be found in \cite[Section 1]{CTSSD87} and \cite[Chapter XIII]{HP94}.
For completeness, we provide a proof of this proposition using primary decomposition and Hilbert-Samuel multiplicity of an ideal. 
In the former, the authors study the cycle decomposition of the intersection of two quadric hypersurfaces to obtain results about existence of rational points on intersection of two quadric hypersurfaces and Ch\^atelet surfaces over number fields. 
Our statements here are slightly simpler to state (and slightly different) since in the works above the authors work in the setting of perfect fields, whereas we are concerned with the special case of algebraically closed fields of characteristic zero.

\subsection{High-level ideas of the proof of Theorem~\ref{thm:deltasg2}}\label{sec:proof-overview}

Suppose we are given a $\drsg{\delta}{2}$ configuration $\cF = \cF_1 \cup \cF_2$, where $\cF_d$ is the set of forms of degree $d$ in our configuration.

Our strategy to prove the robust radical SG theorem is based on the following toy example: suppose our polynomial ring is $\bC[x_1, \ldots, x_r, y_1, \ldots, y_s]$, where one should think of $s$ being constant and $r \gg s$, and every quadratic form $Q$ in our configuration is a polynomial which is ``univariate'' over the smaller polynomial ring $\bC[y_1, \ldots, y_s]$.
That is, for each quadratic $Q$, there exists a linear form $x_Q \in \Cspan{x_1, \ldots, x_r}$ such that $Q \in \bC[x_Q, y_1, \ldots, y_s]$.
In this case, one would hope that the \emph{non-linear} SG dependencies involving our configuration $\cF$ would imply \emph{linear} SG dependencies for the set of linear forms $\cF_1 \cup \{ x_Q \ \mid \ Q \in \cF_2 \}$. 
If we manage to prove that the latter set of linear forms is a \emph{robust linear SG configuration}, we can invoke the robust SG theorem for linear forms of \cite{BDWY11, DSW14}.

In general it is not always possible to reduce the general robust radical SG problem for quadratics to the toy example above.\footnote{For instance if the polynomial $Q = \sum_{i=1}^s x_i y_i$ is in our SG configuration.}
However we will be able to construct a small subalgebra of our polynomial ring which is just as good as the small polynomial ring $\bC[y_1, \ldots, y_s]$ in the toy example above. 
Additionally, we will not always be able to reduce the non-linear problem to a robust linear SG configuration, as some forms $x_Q$ may appear with multiplicity. Instead of a robust linear SG configuration, we will reduce it to a $\delta$-LCC configuration of \cite{BDWY11}.

Since the main counterexample to the above toy example are quadratics of large rank, the small subalgebras that we construct will have both linear and quadratic forms as generating elements.
Therefore, it is natural to consider the vector space of forms generating the algebra, which we denote by $V := V_1 + V_2$, where $V_1$ is the vector space of linear forms in the algebra and $V_2$ is the vector space of quadratic \emph{generators} of the algebra. 
The main idea here is that the quadratic generators will be composed only of quadratics of high rank, which can essentially be though of as ``free variables.''
As it turns out, intuitively and informally, the only properties that we need from the vector space above are that:
\begin{enumerate}
    \item the quadratics in $V_2$ are ``robust'' against the linear forms in $V_1$. That is, we would like each quadratic in $V_2$ to be of very high rank even if we subtract from it polynomials from the algebra $\bC[V_1]$
    \item $V$ is in a sense ``saturated'' with respect to our configuration $\cF$. That is, there exists no small vector space of linear forms that we can add to $V_1$ that would add many polynomials of $\cF$ to the algebra $\bC[V]$, or ``make them closer to being in $\bC[V]$.''
\end{enumerate}

The first condition ensures that any quadratic from our set $\cF$ which ``depends'' on a form from $V_2$ must be of high rank, while the second condition ensures that there is no trivial way to increase the algebra slightly in order to have more forms from $\cF$ inside of the larger algebra.
We call any vector space which satisfies both conditions above a \emph{clean vector space} with respect to $\cF$.\footnote{In hindsight, this saturation condition resembles \cite[Theorem 7.7]{BDWY11},  where they deal with $\delta$-LCC configurations. 
They show that if no element has many repetitions (the trivial case for them), they are able to non-trivially construct a design matrix which annihilates a constant fraction of the elements in the configuration.}

To construct the subalgebra above, we need to understand in a bit more detail the structure of the radical of an ideal generated by two quadratic forms.
To this end, in \cref{sec:algebraic-geom} we prove \cref{theorem: radical structure}, generalizing the previous structure theorems from \cite{S19, PS20a}. 
Additionally, we also assemble results on the structure of minimal primes of these ideals to construct our algebra.

With \cref{theorem: radical structure} at hand, we proceed in a similar fashion as in~\cite{S19, PS20a} by partitioning the quadratics in our SG configuration into subsets, each satisfying a particular case of our structure theorem.
The four subsets we have are as follows: if $\cF$ is a $\drsg{\delta}{2}$ configuration, we take $\varepsilon = \delta/10$ and define
\begin{enumerate}
    \item $\Fspan$ is the set of quadratics $Q$ which satisfy a ``span dependency'' with at least $\varepsilon$-fraction of the forms.
    That is, there exist many quadratics $F, G \in \cF_2$ such that $G \in \Cspan{Q, F}$. 
    \item $\Flinear$ is the set of quadratics $Q$ which satisfy case 3 (c) of \cref{theorem: radical structure} with at least an $\varepsilon$-fraction of the other forms. 
    That is, there are many quadratics $F \in \cF$ and linear forms $x, y$ such that $(Q, F) \subset (x, y)$, and this minimal prime has multiplicity $\geq 2$.
    \item $\Fdeg$ is the set of quadratics $Q$ with an $\varepsilon$-fraction of its SG dependencies with linear forms.\footnote{Deg stands for the degenerate case.}
    \item $\Fsquare$ is the set of quadratics $Q$ which satisfy case 3 (b) of \cref{theorem: radical structure} with at least $(\delta - 3\varepsilon)$-fraction of the other forms. 
    That is, there are many quadratics $F \in \cF$ such that there is a linear form $\ell$ such that $\ell^2 \in \Cspan{F, Q}$.
\end{enumerate}

With this partition, we proceed to construct a small clean vector space $V$ such that $\Fsquare$ and $\Flinear$ are entirely contained in the algebra $\bC[V]$, and the forms in the remaining subsets are either in $\bC[V]$, or are univariate over $\bC[V]$.
Here, by univariate over $\bC[V]$, we mean that there is a linear form $z \not\in V_1$ such that the form is in the algebra $\bC[V][z]$.

We construct the subalgebra above in four steps, where in each step we construct intermediate subalgebras which handle one of the subsets of quadratics defined above. 
To construct the intermediate subalgebras, we use two strategies: iterative processes similar to the ones used in \cite{S19, PS20a}, and constructing double covers of the SG dependencies.
These two strategies allow us to construct algebras generated by $\poly{1/\delta}$ many elements with the desired properties for each of the subsets above.
The iterative processes allow us to tightly control $\Fsquare$ and obtain some control over $\Flinear$ and $\Fspan$, whereas the double covers allow us to handle $\Fdeg$ and also to prove that the remaining linear forms will become a $\delta$-LCC configuration.

Once we have our clean subalgebra with respect to the SG configuration, and we have every polynomial in the configuration either in the algebra or univariate over our algebra, we proceed to prove that the ``additional linear forms'' that arise in this way, together with the linear forms from our configuration, span a vector space of small dimension.
While these linear forms satisfy linear relations, the linear forms corresponding to different quadratics in our set might be the same, and therefore the set of linear forms might not form a robust linear Sylvester-Gallai configuration.
However, the fact that the vector space $V$ is saturated implies that not too many quadratics have the same linear form: if they did, then we could add that linear form to $V_{1}$ and add many polynomials of $\cF$ to $\bC\bs{V}$.
This saturation allows us to show that the linear forms make up a $\delta$-LCC configuration, and therefore span a vector space of small dimension, concluding our proof.

\subsection{Related Work}

The original motivation for studying higher degree SG configurations comes from \cite{gupta2014algebraic}, in order to give polynomial time PIT for a special class of depth-4 algebraic circuits. 
The most general SG problem/configuration that is needed towards this application is the following conjecture, \cite[Conjecture 1]{gupta2014algebraic}, which we term as $(k, d, c)$-Sylvester-Gallai conjecture.

\begin{conjecture}[$(k, d, c)$-Sylvester-Gallai conjecture]
\label{conjecture:dsg}
Let $k, d, c \in \bN^*$ be parameters, and let $\cF_{1}, \dots, \cF_{k}$ be finite sets of irreducible polynomials of degree at most $d$ such that 
\begin{itemize}
    \item $\cap_{i} \cF_{i} = \emptyset$,
    \item for every $Q_{1}, \dots, Q_{k-1}$ each from a distinct set $\cF_{i_{j}}$, there are polynomials $P_{1}, \dots, P_{c}$ in the remaining set such that $\prod P_{i} \in \radideal{Q_{1}, \dots, Q_{k-1}}$.
\end{itemize}
Then the transcendence degree of the union $\cup_{i} \cF_{i}$ is a function of $k, d, c$, independent of the number of variables or the size of the sets $\cF_{i}$.
\end{conjecture}

As a step towards the proof of \cref{conjecture:dsg}, \cite{S19} studies quadratic SG configurations (\cref{conjecture:firstdsg}).
The configurations we study are exactly the fractional versions of these quadratic SG configurations.
In \cite{PS20a}, the authors extend the result on quadratic SG configurations, weakening the SG condition, only requiring that the radical of the ideal generated by every pair of quadratics contains a product of four other quadratics.
In \cite{PS20b}, the authors extend this further, by proving \cref{conjecture:dsg} for the case of $k = 3, d = 2$ and $c = 4$, which gives a polynomial time blackbox PIT for algebraic circuits computing a sum of three products of quadratic polynomials.

Our proof techniques and intermediate results generalise some of those of \cite{S19}, \cite{PS20a}, \cite{PS20b}.
In \cite{S19}, the author proves a structural result for quadratic forms contained in the radical of the ideal generated by two other quadratic forms.
In \cite{PS20a}, this result is extended to products of quadratic forms.
Our structure result directly classifies the radical of the ideal generated by two quadratics based on the number and degree of the minimal primes of the ideal.
Both structure theorems of \cite{S19} and \cite{PS20a} follow as immediate corollaries.

Further, our definition of clean vector spaces and the clean up procedure is a generalisation of part of the strategy in the above works.
In \cite{S19, PS20a}, the authors construct two vector spaces: one of linear forms and another of quadratic forms, and then they prove that most polynomials in the configuration can be written as the sum of a quadratic polynomial in the second vector space, and a polynomial ``close" to the algebra generated by the first vector space.
Our definition of clean vector spaces formalizes this strategy, giving us more structure which helps us unify the case analysis in these previous works.

Another important point to notice is that in this paper we do not make use of the projection trick used in \cite{S19, PS20a}. 
While the parameters become slightly worse for not using the projection trick, as we now have to account for repetitions in the set of linear forms not in the algebra. 
We believe that getting rid of the projection trick will make this strategy more amenable to generalizations to higher degree.

\textbf{Progress on Polynomial Identity Testing:} Recently, there has been remarkable progress on the PIT problem for depth 4 circuits (the same algebraic circuits considered in \cite{gupta2014algebraic}). In \cite{DDS21}, the authors give a  quasi-polynomial time algorithm for blackbox PIT for depth 4 circuits with bounded top and bottom fanins.
Their approach involves considering the logarithmic derivative of circuits, and is analytic in nature, which allows them to bypass the need of Sylvester-Gallai configurations.
Another PIT result in this setting comes from the lower bound against low depth algebraic circuits proved by \cite{LimayeST21}, which gives a weakly-exponential algorithm for PIT for these circuits via the hardness vs randomness paradigm for constant depth circuits \cite{CKS19}.
However, the SG approach of \cite{gupta2014algebraic} is the only one so far which could yield polynomial-time blackbox PIT algorithms for the subclass of depth-4 circuits with constant top and bottom fanins.

\paragraph{Comparison with \cite{OS21}:} In \cite{OS21} the authors generalize the radical SG theorem to cubic forms. However, the SG theorem in \cite{OS21} is not robust.
Robustness significantly increases the combinatorial complexity of the problem (as it is the case in every setting, even in the linear case).

Some ideas in this paper are motivated by similar constructions in \cite{OS21}.
More precisely, the construction of wide algebras motivated our construction of clean vector spaces (\cref{sec:clean-ideals}). 
Wide algebras have stronger algebro-geometric properties and are significantly larger than the algebras generated by  clean vector spaces.
Apart from this motivation, both works are distinct in their techniques, since in our case the robustness severely constrains our choice of dependencies.

\paragraph{Simultaneous result \cite{PS22}:} Simultaneously and independently from this work, Peleg and Shpilka have also proved that $\drsg{\delta}{2}$ configurations have $\poly{1/\delta}$ dimension.
While the result of \cite{PS22} in its current form works when the configuration only has irreducible quadratics, in our work we also allow linear forms in our configurations.

There are a number of parallels between the methods used in \cite{PS22} and the ones used in our paper.
Both use structure theorems for ideals generated by quadratics, and structure theorems for $(x, y)$-primary ideals.
Further, both results divide the configuration into special sets based on the cases of the structure theorem, and control each of these sets separately.

One key technical difference between our approach and \cite{PS22} is the structure used to control the above sets. 
In \cite{PS22} they use an algebra generated by linear forms and quadratics with the property that linear combinations of quadratics are high rank even after taking quotients with the linear forms.
We define the notion of clean vector spaces, which generate ``special algebras'' which apart from having the above property (what we call robustness) are also saturated in the sense that adding a few linear forms cannot bring too many polynomials in our configuration ``closer'' to the vector space.
We also use the notion of univariate polynomials over clean vector spaces, and prove the existence of a small clean vector space $V$ such that the polynomials in each special set is univariate over $V$.
Once we have such structure, we can assign to each univariate polynomial a linear form $\ell_i$ (the ``extra variable'' from this polynomial), and we then show that the set of linear forms $\{\ell_i\}$ corresponding to each polynomial forms a LCC configuration.

Another key technical difference is that in our work, we \textbf{do not} make use of the \emph{projection method}, as we believe that in higher degrees such method may not be amenable to generalization without generalizing the SG conjectures as well. 
This is one of the main reasons why we can only prove that the univariate polynomials $\ell_{i}$ form a LCC configuration, instead of a robust linear SG configuration.
This, and the fact that we handle linear forms, are the reasons why our bound is worse than the one in \cite{PS22}.

Handling the linear forms presents an extra technical challenge.
The main difficulty arises when a quadratic $Q$ satisfies the SG condition with many linear forms $\ell$,
as there is less structure between $Q, \ell$ and the quadratic in $\radideal{Q, \ell}$ than when the configurations just consist of quadratics.
This lack of structure makes our analysis significantly more intricate.

\subsection{Organization}

In \cref{sec:prelim} we establish some notation and preliminary results.
In \cref{sec:algebraic-geom} we establish the algebro-geometric results we need to prove our structural results on the minimal primes of ideals generated by two quadratics.
In this section, we also show how our results generalize previous structural results on SG configurations.
In \cref{sec:clean-ideals} we define and prove results that we need about one of the main objects that we develop: clean vector spaces. 
In \cref{sec:sg-robust} we prove our main result, that robust radical Sylvester-Gallai configurations must lie in a small dimensional vector space. 
In \cref{sec:conclusion} we conclude and state some open problems and future directions.

\section{Preliminaries}\label{sec:prelim}


In this section, we establish the notation which will be used throughout the paper and some important background which we shall need to prove our claims in the next sections.

\subsection{Notation}

Let $S := \bC\bs{x_0, \cdots, x_n}$ be our polynomial ring, with the standard grading by degree.
We can write $S = \bigoplus_{d \geq 0} S_d$, where $S_d$ is the $\bC$-vector space of forms of degree $d$.
For the rest of this paper, we will use the term \emph{form} to refer to a homogeneous polynomial.
Given any graded vector space $V \subset S$, we use $V_d$ to denote the subspace of forms of degree $d$, that is, $V_d := V \cap S_d$.

\subsection{Rank of quadratic forms}

We now define a notion of the rank of quadratic forms, in accordance to \cite{S19}.

\begin{definition}[Rank of a quadratic form]
  \label{def:rank}
  Let $Q$ be a quadratic form.
  The rank of $Q$, denoted $\Rank Q$ is the smallest $s$ such that we can write $Q = \sum_{i=1}^{s} a_{i} b_{i}$ with $a_{i}, b_{i} \in S_{1}$. If $\Rank Q =s$, then such a decomposition $Q = \sum_{i=1}^{s} a_{i} b_{i}$ with $a_{i}, b_{i} \in S_{1}$ is called a minimal representation of $Q$.
\end{definition}

\begin{remark}
Let $Q=\sum_ia_{ii}x_i^2+\sum_{i<j}2a_{ij}x_ix_j$ be a quadratic form in $S$. Recall that there is an one-to-one correspondence between quadratic forms $Q\in S_2$ and symmetric bilinear forms. Let $M$ be the symmetric matrix corresponding to the symmetric bilinear form of $Q$. Note that the $(i,j)$-the entry of $M$ is given by $a_{ij}$. If $M$ is of rank $r$, then  after a suitable linear change of variables, we can write $Q=x_1^2+\cdots+x_r^2$. Since the rank of a quadratic form is invariant under a linear change of variables, we have   $\Rank(Q)=\lceil r/2\rceil$, if $M$ is of rank $r$.
\end{remark}

In the next sections, we will need to use the following notion of a vector space of a quadratic form, which is a slight modification on the definition first given in \cite{S19}. 
The only modification that we make is that we preserve the quadratic form if its rank is high enough.

\begin{definition}[Vector space of a quadratic form]
  \label{def:lin}
  Let $Q$ be a quadratic form of rank $s$, so that $Q = \sum_{i=1}^{s} a_{i} b_{i}$.
  Define the vector space  
  $\linold{Q} := \Cspan{a_{1}, \dots, a_{s}, b_{1}, \dots, b_{s}}$.
  Define $\lin{Q}$ as:
  $$ \lin{Q} =
  \begin{cases}
  \Cspan{Q}, \ \text{ if } s \geq 5 \\
  \linold{Q}, \ \text{ otherwise.}
  \end{cases}
  $$
\end{definition}

Note that $\lin{Q}$ is always a vector space of $\bigO{1}$ dimension (in fact, it is of dimension at most $10$), while $\linold{Q}$ can have non constant dimension. While a minimal representation $Q = \sum_{i=1}^{s} a_{i} b_{i}$ is not unique, the vector space $\linold{Q}$ is unique and hence well-defined. 
The following lemma, which appears in \cite[Fact 2.15]{PS20a} characterizes $\linold{Q}$ as the smallest vector space of linear forms defining the algebras that contain $Q$.

We also extend the definition of $\Lin$ to linear forms in the natural way as follows.
\begin{definition}
  \label{def:linearformlin}
  For a linear form $\ell \in S_{1}$ define $\lin{\ell} := \Cspan{\ell}$.
\end{definition}

\begin{lemma}\label{lem:linold}
If $Q = \sum_{i=1}^r x_{i} y_{i}$ with $x_{i}, y_{i} \in S_{1}$ then $\linold{Q} \subseteq \Cspan{x_{i}, y_{j}|i,j\in[r]}$.
\end{lemma}

\begin{proposition}\label{proposition:factor-in-pencil}
   Let $Q_1,Q_2\in S_2$. Suppose $\Rank(Q_1)\geq d$ for some $d$. Consider the pencil of quadratic forms formed by $Q_1,Q_2$, i.e. the set $\{Q_t=Q_1+tQ_2|t\in \bC\}$. Then there are at most $2d-1$ values of $t$ for which the quadratic form $Q_t$ has $\Rank(Q_t)<d$.
\end{proposition}
\begin{proof}
    Let $M_1,M_2$ be the matrices corresponding to the quadratic forms $Q_1,Q_2$ respectively. After a linear change of variables, we may assume that the matrix $M_2$ is a diagonal matrix. Then the matrix corresponding to $Q_t=Q_1+tQ_2$ is given by $M_t=M_1+tM_2$. Note that $\Rank(Q_t)<d$ if and only if $\Rank(M_t)<2d-1$. Now $\Rank(M_t)<2d-1$ if and only if all the determinants of all the $(2d-1)\times (2d-1)$ minors are zero. Note that each such determinant is a polynomial in $t$. Since $\Rank(Q_1)\geq d$, there exists at least one such determinant that is non-zero at $t=0$, and hence it is a non-zero polynomial in $t$ of degree $\leq 2d-1$. Thus, there are at most $2d-1$ values of $t$ such that  $\Rank(Q_t)<d$.
\end{proof}

\begin{remark}\label{remark:factor-in pencil}
    If $Q_1, Q_2$ are irreducible quadratic forms which satisfy the conditions of the previous proposition with $d=2$ and there are at least two values of $t$ for which $\Rank(Q_t) = 1$, then $\Rank(Q_1) = \Rank(Q_2) = 2$. Further we have $\lin{Q_1} = \lin{Q_2}$, since both $Q_{1}$ and $Q_{2}$ will be linear combinations of the two reducible polynomials in the pencil defined by $Q_{1}, Q_{2}$.
\end{remark}

\begin{lemma}\label{lemma:almost-minimal-collapse-poly-ring}
If $V_1 \subset S_1$ and $x, y, u, v \in S_1$ such that $xy - uv \in \bC[V_1]$ then one of the following holds:
\begin{enumerate}
    \item We have $x, y, u, v \in V_1$.
    \item There exists a linear form $\ell \in S_1 \setminus V_1$ such that $x,y,u,v\in \Cspan{V_1,\ell}$.
\end{enumerate}
\end{lemma}

\begin{proof}
Let $P := xy-uv$.
If $P$ is irreducible, then $xy - uv$ is a minimal representation of $P$. Therefore, we have $\linold{P}=\Cspan{x,y,u,v}$ and \cref{lem:linold} implies $x,y,u,v\in V_1$. 

Thus, we can assume $P$ factors and $xy-uv=zw$ for some $z,w\in S_1$. 
Since $P=zw$ is a minimal representation of $P$, \cref{lem:linold} implies $\linold{P}= \Cspan{z, w}\subset V_1$. 
If $x, y, u, v \not\in V_1$, we can assume w.l.o.g. that $x \not\in V_1$, and therefore $zw \neq 0$ in $S/(x)$.
By factoriality of $S/(x)$, we have that $-uv \equiv zw \neq 0 \bmod (x)$ which implies $u, v \in \Cspan{x, V_1}$.
Since $xy = uv + zw$, \cref{lem:linold} implies that $\linold{xy} = \Cspan{x, y} \subset \Cspan{u,v, z,w} \subset \Cspan{x, V_1}$ as we wanted.
\end{proof}

\begin{proposition}\label{proposition:irreducible-radical-lin}
  Let $Q \in S_2$ be an irreducible quadratic form. 
  If there exists a linear form $\ell$ such that $(Q,\ell)$ is not radical, then $\Rank Q = 2$, $\dim \lin{Q}=3$, and any such $\ell$ must be in $\lin{Q}$. 
  Moreover, any three pairwise linearly independent linear forms $\ell_1, \ell_2, \ell_3$ such that $(Q, \ell_i)$ is not radical are such that 
  $\lin{Q} = \Cspan{\ell_1, \ell_2, \ell_3}$.
\end{proposition}

\begin{proof} 
Since $Q$ is irreducible, we have $\Rank(Q)\geq 2$ and $\dim \lin{Q} \geq 3$.
If $\Rank(Q) \geq 3$ then $Q$ is irreducible modulo $\ell$, for any linear form $\ell$. Therefore the ideal ideal $(Q, \ell)$ is prime, and hence radical for any linear form $\ell$. Hence if there exists a linear form $\ell$ such that $(Q,\ell)$ is not radical, then we must have that $\Rank(Q) = 2$. 
If $Q$ is reduced modulo a linear form $\ell$, then the ideal $(Q,\ell)$ is radical. 
Therefore, if $(Q,\ell)$ is not a radical ideal for some linear form $\ell$, then we have $Q = \ell a + b^2$, for some $a, b \in S_1$. Note that this is a minimal representation of $Q$. Therefore, by \cite[Fact 2.13]{PS20a} we must have that $\Cspan{\ell,a,b}= \lin{Q}$. Hence we have $\ell\in\lin{Q}$. 
Therefore $\mathrm{dim}(\lin{Q})\leq 3$, which together with $\dim \lin{Q} \geq 3$ implies $\dim \lin{Q} = 3$.

Suppose there exists a linear form $x$ such that $(Q,x)$ is not radical. 
Then, by the above discussion, we may assume that $Q=xy+z^2$ where $x, y,z$ are linearly independent forms and $\lin{Q}=\Cspan{x,y,z}$.  
Suppose $(Q,\ell_i)$ is not radical for three pairwise linearly independent forms $\ell_1,\ell_2,\ell_3\in S_1$. 
Consider the projective conic defined by $Q$ in $\bP^2$. 
Now $(Q,\ell_i)$ not radical implies that the line defined by $\ell_i$ is tangent to the conic $Q$ at some point $p_i$. 
Since $\ell_1,\ell_2,\ell_3$ are pairwise linearly independent, the points $p_1,p_2,p_3$ are distinct. 
If $\ell_1,\ell_2,\ell_3$ are linearly dependent then the points $p_1,p_2,p_3$ are collinear. 
Then there exists a line passing through $p_1,p_2,p_3$ which intersects the projective conic $Q$ at three distinct points. This contradicts Bezout's theorem, and hence $\ell_1,\ell_2,\ell_3$ are linearly independent. 
\end{proof}

\subsection{Sylvester-Gallai Configurations and LCC Configurations}\label{sec:sg-configurations}

We now establish the notation we use for the SG configurations that we will encounter, as well as what is known about the bounds on their dimension.
We also define the notion of a $\delta$-LCC configuration introduced in \cite{BDWY11}, and state a bound on the rank of such configurations.

We start by restating the definition of quadratic robust Sylvester-Gallai configurations.

\deltasgconfig*

In the above definition, if $F_{k} \in \cF$ is such that $F_{k} \in \radideal{F_{i}, F_{j}}$, then we refer to $(F_{i}, F_{j}, F_{k})$ as a SG-triple. 
Note that the radical condition is not symmetric, i.e. if $(F_i,F_j,F_k)$ is a SG-triple, it is not necessary that $(F_i,F_k,F_j)$ or $(F_k,F_j,F_i)$ is a SG-triple.
The general form of the linear SG configurations that we will encounter are as follows:

\begin{remark}
If $\cF$ is a $\drsg{\delta}{2}$ configuration where $\cF_2 = \emptyset$, we obtain the robust linear SG configurations from \cite{BDWY11, DSW14}.
Hence, in this case we shall simply denote the robust linear SG configurations as $\delta$-linear SG configurations.
\end{remark}

\begin{theorem}[Rank Bound for Linear SG Configurations \cite{DSW14}]\label{thm:DSW14}
  If $\cF$ be a $\delta$-linear SG configuration, then $\dim \br{\Kspan{\cF}} \leq 13 /\delta$.
\end{theorem}

We now define $\delta$-LCC configurations.
These are configurations that arise in the study of locally correctable codes, and hence the name.
They can be seen as extensions of SG configurations.
\begin{definition}[$\delta$-LCC configurations, \cite{BDWY11}]
Let $0 < \delta \leq 1$ and let $\cF := \bc{x_{1}, \dots, x_{m}} \subset S_1$ be a multiset of linear forms, not necessarily distinct.
We say that $\cF$ is a $\delta$-LCC configuration, if for every $i \in \bs{m}$ and every subset $\Gamma \subset \bs{m}$ with $|\Gamma| \leq \delta m$, there are indices $j, k \in \bs{m} \setminus \Gamma$ such that either $x_{i} \in \Cspan{x_j} \cup \Cspan{x_k}$ or the three linear forms $x_{i}, x_{j}, x_{k}$ are pairwise independent and satisfy $x_{i} \in \Cspan{x_{j}, x_{k}}$.
\end{definition}

In \cite[Theorem~7.6]{BDWY11} the authors proved the following bound for $\delta$-LCCs.

\begin{theorem}\label{thm:BDWY11}
If $\cF$ is a $\delta$-LCC configuration then $\dim \br{\Cspan{\cF}} = \bigO{1 / \delta^{9}}$.
\end{theorem}

\section{Structural Algebro Geometric Results}\label{sec:algebraic-geom}


In this section we establish the necessary definitions and theorems needed from commutative algebra and algebraic geometry.
Throughout this section, we will work over an algebraically closed field of characheristic zero $\bK$ and we will denote $S=\bK\bs{x_0, x_1,\cdots ,x_n}$.

\subsection{Algebraic preliminaries}
All rings that we consider in this section are commutative rings with unity. 
We refer to \cite{AM69, Eis95} for general background on commutative algebra, in particular for primary decomposition, Cohen-Macaulay rings and Hilbert-Samuel multiplicities. 
We briefly recall some of the relevant statements below for convenience.

\begin{definition}[Regular sequence]\label{definition: regular sequence}
Let $R$ be ring and $M$ an $R$-module. A sequence of elements $x_1, x_2, \cdots x_n \in R$ is called an $M$-regular sequence if
\begin{itemize}
\item[(1)] $(x_1,x_2,\cdots,x_n)M\neq M$, and
\item[(2)] for each $i$, the element $x_i$ is a non-zerodivisor in $M/(x_1,\cdots,x_{i-1})M$.
\end{itemize}   
If $M=R$, then we simply call $x_1, \ldots, x_n$ a regular sequence.
\end{definition}

\begin{remark}\label{remark: regular sequence}
Let $Q_1, Q_2 \in S_2$.
If $Q_1,Q_2$ do not have any common factors, then $Q_1,Q_2$ is a regular sequence.
Indeed, if $Q_1, Q_2$ is not a regular sequence then $Q_2$ is a zero-divisor in $S/(Q_1)$. 
Therefore,  $PQ_2=RQ_1$ for some forms $P\not\in (Q_1)$ and $R\in S$. Note that $Q_1$ does not divide $Q_2$. Since $S$ is an UFD, and $Q_1,Q_2$ do not have any common factors, we conclude that $Q_1$ divides $P$, which is a contradiction. Therefore $Q_1,Q_2$ is a regular sequence.
\end{remark}

The following proposition tells us that ideals generated by $Q_1,Q_2$ are Cohen-Macaulay, if $Q_1,Q_2$ do not have a common factor. This is a special case of \cite[Proposition 18.13]{Eis95}, which states that complete intersections are Cohen-Macaulay.

\begin{proposition}[Cohen-Macaulayness \cite{Eis95}]\label{prop:cohen-macaulay}
Let $Q_1, Q_2 \in S_2$ such that $Q_1,Q_2$ do not have any common factors, then the ideal $(Q_1, Q_2)$ is Cohen-Macaulay, i.e. the quotient ring $S/(Q_1,Q_2)$ is a Cohen-Macaulay ring.
\end{proposition}
 
\begin{proof}
By Remark \ref{remark: regular sequence}, we know that $Q_1,Q_2$ is a regular sequence. 
Therefore co-dimension of the ideal $(Q_1,Q_2)$ is $2$. 
Since the polynomial ring $S$ is Cohen-Macaulay, we conclude that $S/(Q_1,Q_2)$ is Cohen-Macaulay by \cite[Proposition 18.13]{Eis95}. 
\end{proof}

Cohen-Macaulay ideals are equidimensional and unmixed. 
Hence, if $I$ is a Cohen-Macaulay ideal then every associated prime of $I$ is a minimal prime and the codimension of every minimal prime of $I$ is the same  \cite[Corollaries 18.11, 18.14]{Eis95}. An easy consequence is the following result.

\begin{corollary}[Unmixedness and Equidimensionality]\label{cor:unmixedness}
    If $Q_1, Q_2 \in S_2$ without any common factors, then the ideal $(Q_1, Q_2)$ is unmixed and equidimensional. In particular, let $(Q_1,Q_2)= \mathfrak{q}_1\cap\cdots\cap \mathfrak{q}_n$ be an irredundant primary decomposition. Let $\mathfrak{p}_i=\mathrm{rad}(\mathfrak{q}_i)$. Then
    \begin{enumerate}
        \item every associated prime of $(Q_1,Q_2)$ is a minimal prime, i.e. $\mathfrak{p}_1, \dots, \mathfrak{p}_n$ are precisely the distinct minimal primes over $(Q_1,Q_2)$,
        \item every minimal prime of $(Q_1, Q_2)$ has codimension 2, i.e. $\mathrm{ht}(\mathfrak{p}_i)=2$ for all $i\in[n]$.
    \end{enumerate}  
\end{corollary}
\begin{proof}
    Note that by Proposition \ref{prop:cohen-macaulay}, the ideal $(Q_1,Q_2)$ is Cohen-Macaulay. Hence $(Q_1,Q_2)$ is unmixed and equidimensional.
\end{proof}

\noindent Let $(R,\mathfrak{m})$ be a local ring and $M$ be an $R$-module. 
The Hilbert-Samuel function of $M$ is defined as
\[H_{\mathfrak{m},M}(n)=\mathrm{length}\br{\frac{\mathfrak{m}^nM}{\mathfrak{m}^{n+1}M}}.\]

By \cite[Proposition 12.2, Theorem 12.4]{Eis95}, there exists a polynomial $P_{\mathfrak{m},M}$ of degree $\mathrm{dim}(M)-1$ such that $P_{\mathfrak{m},M}=H_{\mathfrak{m},M}(n)$ for $n \gg 0$. The polynomial $P_{\mathfrak{m},M}$ is called the Hilbert-Samuel polynomial of $M$. Let $a_d$ be the leading coefficient of $P_{\mathfrak{m},M}$, where $d=\mathrm{dim}(M)-1$. The Hilbert-Samuel multiplicity of $M$ is defined as

\[e(\mathfrak{m},M)=(d-1)!a_d.\]
Therefore, the leading coefficient of the Hilbert-Samuel polynomial $P_{\mathfrak{m},M}$ is $\frac{e(\mathfrak{m},M)}{(\mathrm{dim}M-1)!}$.\\

Let $R=S_{\mathfrak{m}}$ be the localization of $S$ at the irrelevant maximal ideal $\mathfrak{m}=(x_0,\cdots,x_n)$. 
Let $I$ be a homogeneous ideal in $S$. 
Then the localization $(S/I)_{\mathfrak{m}}$ is an $R$-module. 
We will denote $e(S/I):=e(\mathfrak{m},(S/I)_{\mathfrak{m}})$. Note that by \cite[Exercise 12.6]{Eis95}, the number $e(S/I)$ is also equal to the degree of the projective variety defined by $I$ in $\mathbb{P}^n$.

Let $I=\mathfrak{q}_1\cap\cdots\cap\mathfrak{q}_m$ be an irredundant primary decomposition of $I$ in $S$. 
Let $\mathfrak{p}_i=\radideal{\mathfrak{q}_i}$ be a minimal prime of $I$ for some $i$. 
Then the localization $(S/I)_{\mathfrak{p}_i}$ is an $S_{\mathfrak{p}_i}$ module of finite length. 
We define the multiplicity of $\mathfrak{p}_i$ in the primary decomposition of $I$ as $$m(\mathfrak{p}_i)=\mathrm{length}((S/I)_{\mathfrak{p}_i}).$$ 
\begin{remark}\label{remark: length minimal prime}
Note that $m(\mathfrak{p}_i)\geq 1$.  
If $m(\mathfrak{p}_i)=1$, then we must have $\mathfrak{q}_i=\mathfrak{p}_i$. 
Indeed, $\mathrm{length}((S/\mathfrak{q}_i)_{\mathfrak{p}_i})\leq \mathrm{length}((S/I)_{\mathfrak{p}_i})$ since we have a surjective homomorphism $(S/I)_{\mathfrak{p}_i}\rightarrow (S/\mathfrak{q}_i)_{\mathfrak{p}_i}$. 
Therefore we have $\mathrm{length}((S/\mathfrak{q}_i)_{\mathfrak{p}_i})=1$. 
If $\mathfrak{q}_i\subsetneq \mathfrak{p}_i$, then we have a strict chain of $S_{\mathfrak{p}_i}$-modules given by $(0)\subsetneq (\mathfrak{p}_i)_{\mathfrak{p}_i}\subsetneq (S/\mathfrak{q}_i)_{\mathfrak{p}_i} $ , which is of length $2$. That is a contradiction. Therefore we must have $\mathfrak{q}_i= \mathfrak{p}_i$. Hence if we have $m(\mathfrak{p}_j)=1$ for all $j$, then $I=\cap_j\mathfrak{p}_j$, and $I$ is a radical ideal.

Conversely if $I$ is a radical ideal then $m(\mathfrak{p}_i)=1$ for all minimal primes.
Indeed for every $i$ the ideal $\kp_{i} \br{R / I}_{\kp_{i}}$ is the nilradical of $\br{R / I}_{\kp_{i}}$ by the minimality of $\kp_{i}$.
Since $R/I$ has no nilpotents, we have $\kp_{i} \br{R / I}_{\kp_{i}} = \ideal{0}_{\kp_i}$, which implies $\br{R / I}_{\kp_{i}}$ is a field, and therefore has length $1$.
\end{remark}

We recall the basic properties of the Hilbert-Samuel multiplicity below.

\begin{proposition}\label{proposition: HS multiplicity}\cite[Exercises 12.7, 12.11]{Eis95}
  Let $I \subset S$ be a homogeneous ideal.
  \begin{enumerate}
      \item Let $J\subset S$ be a homogeneous ideal such that $I\subset J$. Then $e(S/J)\leq e(S/I)$.
      \item If $I=(F)$ for some homogeneous polynomial of degree $d$, then $e(S/I)=d$.
      \item If $I=(F_1,\cdots,F_m)$ where $F_1,\cdots,F_m$ is a regular sequence of homogeneous polynomials and $\mathrm{deg}(F_i)=d_i$. Then $e(S/I)=d_1\cdots d_m$.
      \item If $I=(F_1,\cdots,F_m)$ where $F_1,\cdots,F_m$ is a regular sequence of homogeneous polynomials, and $\mathfrak{p}_1,\cdots,\mathfrak{p}_m$ are the minimal primes of $I$ in $S$. Then
      \[e(S/I)=\sum_im(\mathfrak{p}_i)e(S/\mathfrak{p}_i).\]
  \end{enumerate}
\end{proposition}

\subsection{Structural Results for Ideals generated by two quadratics}

In this section we prove some results on the structure of non-radical ideals and their primary decompositions, which will be used in our proof of the robust Sylvester-Gallai theorem.

\begin{proposition}\label{proposition:minimal primes}
  Let $Q_1,Q_2\in S_2$ be such that they do not have any common factors and $I = (Q_1,Q_2)$. Let $I=\mathfrak{q}_1\cap\cdots\cap\mathfrak{q}_m$ be an irredundant primary decomposition and $\mathfrak{p}_i=\radideal{\mathfrak{q}_i}$. Then 
  \begin{enumerate}
      \item We have $e(S/\mathfrak{q_i})\leq 4$ for all $i\in [m]$.
      \item If $I$ is not a radical ideal, then $e(S/\mathfrak{p}_i)\leq 2$ for all $i\in [m]$.
      \item If $I$ is not a radical ideal, then for all $i\in [m]$, the minimal primes $\mathfrak{p}_i$ are of one of the following types :
      \begin{itemize}
          \item [(a)] (Linear prime) we have $\mathfrak{p}_i=(x,y)$, for some linearly independent forms $x,y\in S_1$ 
      \item[(b)] (Quadratic prime) we have $\mathfrak{p}_i=(x,Q)$, for some $x\in S_1$ and $Q\in S_2$ such that $Q$ is irreducible modulo $x$.
      \end{itemize}
      \item If $I$ is a radical ideal that is not prime, then either $I$ has a minimal prime that is a linear prime, or $I$ has two minimal primes, both of which are quadratic primes.
  \end{enumerate}
\end{proposition}
\begin{proof}
    Since $Q_1,Q_2$ is a regular sequence we know that $e(S/I)=4$, by Proposition \ref{proposition: HS multiplicity}. Since $I$ is a homogeneous ideal, the primary ideals $\mathfrak{q}_i$ and the prime ideals $\mathfrak{p}_i$ are all homogeneous. Therefore by Proposition \ref{proposition: HS multiplicity}, we see that $e(S/\mathfrak{q}_i)\leq 4$ for all $i$.
    By Proposition \ref{proposition: HS multiplicity}, we also have $\sum_im(\mathfrak{p}_i)e(S/\mathfrak{p}_i)=4$. \\
    
    Suppose first that $I$ is not radical.
    We must have $\mathfrak{q}_j \subsetneq \mathfrak{p}_j$ for some $j$.
    Therefore $m(\mathfrak{p}_j)>1$ for some $j$.
    Since $e(S/\mathfrak{p}_i)\geq 1$ for all $i$, we see that $e(S/\mathfrak{p}_i)\leq 2$ for all $i$. \\
    
    Note that $\mathrm{ht}(\mathfrak{p}_i)=2$ for all $i$, by Corollary \ref{cor:unmixedness}. By \cite[Page 112]{EG84}, any height $2$ prime ideal $\mathfrak{p}$ in $S$ with multiplicity $e(S/\mathfrak{p})=1$ is of the form $(x,y)$ where $x,y\in S_1$ are linearly independent. By \cite[Proposition 11]{Eng07}, any height $2$ prime ideal $\mathfrak{p}$ in $S$ with $e(S/\mathfrak{p})=2$ is of the form $(x,Q)$ where $x\in S_1$, $Q\in S_2$ and $Q$ is irreducible modulo $x$.
    
    Suppose next that $I$ is radical but not prime, so every $m(\kp_{i}) = 1$ and the number of primes $m$ is greater than $1$.
    Either there is some $\kp_{i}$ such that $e(S / \kp_{i}) = 1$ or we have $m = 2$ and $e(S / \kp_{1}) = e(S / \kp_{2}) = 2$.
    In the first case, the ideal $I$ has a linear minimal prime, and in the second case the ideal $I$ has two quadratic minimal primes as claimed.
\end{proof}

As an application we obtain the following  structural result for ideals generated by two quadratic forms. 
This result is a generalization of the structural results in \cite{S19, PS20a}.
In our result, in addition to providing a new proof of the results in  \cite{S19, PS20a}, we obtain information about the minimal primes of the ideal $(Q_1,Q_2)$ as well. 

\radicalstructurethm*

\begin{proof}
Suppose $\ideal{Q_{1}, Q_{2}}$ is radical but not prime.
If $Q_{1}$ is reducible, then it must factor as a product of two linearly independent forms.
The ideal then satisfies case 2.(a).
The same occurs if $Q_{2}$ factors, and therefore we can assume that both $Q_{1}, Q_{2}$ are irreducible.
By \cref{proposition:minimal primes}, either $\ideal{Q_{1}, Q_{2}}$ has a linear minimal prime, or we have $\ideal{Q_{1}, Q_{2}} = \ideal{P_{1}, \ell_{1}} \cap \ideal{P_{2}, \ell_{2}}$ with $\ideal{P_{i}, \ell_{i}}$ prime and $P_i \in S_2$.
In the first case, the ideal $\ideal{Q_{1}, Q_{2}}$ satisfies 2.(b).
Suppose we are in the second case.
Then $Q_i = \alpha_i P_1 + \ell_1 a_i$ for $\alpha_i \in \bK^*$ and $a_i \in S_1$, which implies $\alpha_2 Q_1 - \alpha_1 Q_2 = \ell_1 (\alpha_2 a_1 - \alpha_1 a_2)$.
Note that $\alpha_2 a_1 - \alpha_1 a_2 \not\in (\ell_1)$, otherwise $(Q_1, Q_2)$ would not be radical. 
Hence we are in case 2.(a).

Suppose $(Q_1,Q_2)$ is not a radical ideal. If $Q_1,Q_2$ have a common factor, then we let $Q_1=xy$, $Q_2=xz$. If $x,y,z$ are linearly independent, then we have $(Q_1,Q_2)=(x)\cap (y,z)$ which is an intersection of prime ideals, hence radical. Since $(Q_1,Q_2)$ is not radical, we must have that $x,y,z$ linearly dependent. Therefore $z=\alpha x+\beta y$ for some $\alpha,\beta \in k$. If $\alpha=0$, then $(Q_1,Q_2)=(xy)$. Since $(Q_1,Q_2)$ is not radical, we must have $x,y$ linearly dependent and then $x^2\in \mathrm{span}(Q_1,Q_2)$. If $\alpha\neq 0$, then we see that $x^2\in (Q_1,Q_2)$. Therefore we see that 3.(a) occurs in this case. 
    
Suppose that $Q_1,Q_2$ do not have a common factor and the ideal $I=(Q_1,Q_2)$ is not a radical ideal. Let $I=\cap_i \mathfrak{q}_i$ be an irredundant primary decomposition and let $\mathfrak{p}_i=\radideal{\mathfrak{q}_i}$. Then by Remark \ref{remark: length minimal prime}, there exists a minimal prime $\mathfrak{p}$ such that the multiplicity $m(\mathfrak{p})\geq 2$. Recall that by Proposition \ref{proposition:minimal primes}, all the minimal primes of $I$ are linear or quadratic. 

Suppose there exists a minimal prime $\mathfrak{p}$ of $(Q_1,Q_2)$, such that $\mathfrak{p}=(x,Q)$ for some $x\in S_1$, $Q\in S_2$ where $Q$ is irreducible modulo $x$ and $m(\mathfrak{p})\geq 2$. By Proposition \ref{proposition: HS multiplicity}, we have $$\sum_im(\mathfrak{p}_i)e(S/\mathfrak{p}_i)=4.$$ Note that $x,Q$ is a regular sequence and hence $e(S/\mathfrak{p})=2$. Therefore $m(\mathfrak{p})\leq 2$. Hence $m(\mathfrak{p})=2$, and we must have that $\mathfrak{p}$ is the only minimal prime and $I$ is $\mathfrak{p}$-primary. Consider the chain of $S_{\mathfrak{p}}$-submodules given by $(0)\subset (x^2)\subsetneq (x)\subsetneq (S/I)_{\mathfrak{p}}$. Since $m(\mathfrak{p})=\mathrm{length}(S/I)_{\mathfrak{p}})=2$, we must have $x^2=0$ in $(S/I)_{\mathfrak{p}}$. Therefore by \cite[Corollary 10.21]{AM69}, we conclude that $x^2\in I$. Since $Q_1,Q_2$ are of degree $2$, we must have that $x^2\in \mathrm{span}(Q_1,Q_2)$. 
In this case, 3.(b) occurs.
  
Therefore we may assume that all quadratic minimal primes have multiplicity $1$. Hence $\mathfrak{q}_i=\mathfrak{p}_i$ for all quadratic primes $\mathfrak{p}_i$. If we have that $\mathfrak{q}_j=\mathfrak{p}_j$ for all linear primes $\mathfrak{p}_j$, then the ideal $I$ would be a radical ideal. Therefore there exists a linear prime $\mathfrak{p}_j$ such that the $\mathfrak{p}_j$-primary ideal $\mathfrak{q}_j$ is properly contained in $\mathfrak{p}_j$. Then $(0)\subsetneq \mathfrak{p}_j\subsetneq (S/\mathfrak{q}_j)_{\mathfrak{p}_j}$ is a strict chain of $S_{\mathfrak{p}_j}$-submodules. Therefore we have $\mathrm{length}((S/\mathfrak{q}_j)_{\mathfrak{p}_j})\geq 2$ and hence, $e(S/\mathfrak{q}_j)=\mathrm{length}((S/\mathfrak{q}_j)_{\mathfrak{p}_j})\cdot e(S/\mathfrak{p}_j)\geq 2$. Thus we see that in case 3, if (a) or (b) do not occur then case (c) must occur.
\end{proof}

As immediate corollaries of \cref{theorem: radical structure}, we obtain the structural results proved in \cite{S19,PS20a}. Note that, unlike in the results of \cite{S19,PS20a}, in our \cref{theorem: radical structure} we do not use the existence of a third quadratic or a product of quadratics in the radical of $(Q_1,Q_2)$.

\begin{corollary}[Theorem~29 in \cite{S19}]
  \label{cor:quadraticstructure}
  Let $Q_{1}, Q_{2}, Q_{3}$ be quadratic forms such that $Q_{3} \in \radideal{Q_{1}, Q_{2}}$.
  One of the following holds:
  \begin{enumerate}
    \item The polynomial $Q_{3}$ is in the linear span of $Q_{1}$ and $Q_{2}$.
    \item There exists a linear form $\ell$ such that $\ell^{2}$ is in the linear span of $Q_{1}$ and $Q_{2}$.
    \item There exist linear forms $x$ and $y$ such that $Q_{1}, Q_{2}, Q_{3}$ all belong to the ideal $\ideal{x, y}$.
  \end{enumerate}
\end{corollary}
\begin{proof}
    If $Q_1,Q_2$ satisfy condition $1$ or $2$ of Theorem \ref{theorem: radical structure}, then we see that $Q_3\in \radideal{Q_1,Q_2}=(Q_1,Q_2)$. Since $Q_1,Q_2,Q_3$ are all degree $2$ forms, we must have $Q_3\in \mathrm{span}(Q_1,Q_2)$. If $Q_1,Q_2$ satisfy conditions $3.(a)$ or $3.(b)$ of Theorem \ref{theorem: radical structure}, then we see that case $2$ above holds. Otherwise, condition $2.(c)$ must hold, and there is a prime ideal $(x,y)$ such that $(Q_1,Q_2)\subset (x,y)$. Then we also have $Q_3\in \radideal{Q_1,Q_2}\subset (x,y)$.
\end{proof}

\begin{corollary}[Theorem 3.1 in \cite{PS20a}]
  Let $Q_{1}, Q_{2}, P_1,\cdots P_k$ be quadratic forms such that $\Pi_iP_i\in \radideal{Q_1,Q_2}$. Then one of the following cases holds:
  \begin{enumerate}
     \item There exists $i\in [k]$ such that $P_i\in \mathrm{span}(Q_1,Q_2)$. \item There exists a non-trivial linear combination of the form $\alpha Q_1+\beta Q_2=xy$ where $x,y$ are linearly independent linear forms. 
  \item There exist two linear forms $x,y$ such that $Q_1,Q_2, \Pi_iP_i\in (x,y)$.
  \end{enumerate}
\end{corollary}

\begin{proof}
    If $Q_1,Q_2$ satisfy condition $1$ of Theorem \ref{theorem: radical structure}, then $(Q_1,Q_2)$ is a prime ideal. Hence we have $P_i\in (Q_1,Q_2)$ for some $i$. Since $P_i$ is of degree $2$, we see that in this case we must have $P_i\in \mathrm{span}(Q_1,Q_2)$. If $Q_1,Q_2$ satisfy one of the conditions $2.(a)$, $3.(a)$ or $3.(b)$, then we see that case $2$ above holds. Otherwise, $Q_1,Q_2$ must satisfy one of conditions $2.(b)$ or $3.(c)$, then we see that case $3$ above occurs.
\end{proof}

We now prove some structural results for quadratics with linear minimal primes.

\begin{proposition}[Structure of $(x,y)$-primary ideals of $(Q_1, Q_2)$]\label{proposition:xy-primary}
  Let $Q_1,Q_2\in S_2$ be two irreducible forms such that $\lin{Q_1} \neq \lin{Q_2}$. Suppose $(Q_1, Q_2)$ is not a radical ideal and $\mathrm{span}(Q_1,Q_2)$ does not contain the square of any linear form. Then the following holds.
  
  \begin{enumerate}
      \item There exists a linear minimal prime $(x,y)$ of $(Q_1,Q_2)$ where $x, y \in S_1$ are linearly independent forms, and $(x, y^2)$ is the $(x,y)$-primary component of $(Q_1, Q_2)$.
      \item There exist $a, b \in S_1$ such that $\dim(\Cspan{x, y, a, b}) = 4$ and
  $$ Q_1 = ax + y^2, \ Q_2 = bx + y^2. $$
  \end{enumerate} 
\end{proposition}

\begin{proof}
By our assumption the cases 1, 2.(a),2.(b) of Theorem \ref{theorem: radical structure} cannot occur. 
Therefore case 2.(c) must occur and hence there exists a minimal prime $\mathfrak{p}$ of $(Q_1,Q_2)$, such that $\mathfrak{p}=(x,y)$ for some linearly independent forms $x,y\in S_1$, and the $(x,y)$-primary ideal $\mathfrak{q}$ has multiplicity $e(S/\mathfrak{q})\geq 2$.  
Note that by Proposition \ref{proposition: HS multiplicity}, we have $e(S/\mathfrak{q})\leq 4$.  
By checking the classification of $(x, y)$-primary ideals of multiplicity $\leq 4$ from \cite{MM18}, we see that the only case in which $\mathfrak{q}$ can contain two irreducible quadratics $Q_1, Q_2$ with $\lin{Q_1} \neq \lin{Q_2}$ is case 1 of \cite[Proposition 1.1]{MM18}. 
Thus, after a linear change of variables we must have $\mathfrak{q} = (x, y^2)$ and since $\lin{Q_1} \neq \lin{Q_2}$ we must have (up to scalar multiplication) $Q_1 = ax + y^2$ and $Q_2 = bx + y^2$ where $a,b,x,y$ are linearly independent. This concludes this proof.
\end{proof}

\begin{proposition}\label{proposition:radical-linear}
  If $a, b, x, y \in S_1$ are linearly independent forms and $Q_1 = ax + y^2$, $Q_2 = bx + y^2$, then
  $$ \radideal{Q_1, Q_2} = (Q_1, Q_2, y(a - b)) $$
\end{proposition}

\begin{proof}
We claim that the minimal primes of the ideal $\ideal{Q_{1}, Q_{2}}$ are $\ideal{x, y}$ and $\ideal{Q_{1}, a - b}$.
We first show that these ideals are prime, and then show that they are the minimal primes.
    
That $\ideal{x, y}$ is prime follows from the fact that it is generated by two linear forms.
To show that $\ideal{Q_{1}, a-b}$ is prime, it suffices to show that $Q_{1}$ is irreducible in the ring $S / (a - b)$, since the latter is isomorphic to a polynomial ring.

If $Q_1$ factors in $S/(a-b)$, then $ax + y^{2} = c d$ in $S / (a - b)$, for linear forms $c, d$.
Quotienting further by $c$, we get that $a, x \in (y)$ in $S / (a-b, c)$.
This implies $\Kspan{a, b, x, y} \subset \Kspan{a, x, y, a-b, c} = \Kspan{y,a- b, c}$, contradicting $\dim \Kspan{a,b,x,y} = 4$.
This shows that $\ideal{Q_{1}, a-b}$ is prime.
    
Suppose $P$ is a minimal prime containing $\ideal{Q_{1}, Q_{2}}$.
The ideal $P$ contains $aQ_{2} - b Q_{1} = y^2(a - b)$.
Since $P$ is prime, it must contain either $y$ or $(a - b)$.
If $P$ contains $a - b$, then $P$ clearly contains $\ideal{Q_{1}, a - b}$ and by minimality we must have $P = (Q_1, a-b)$.
If $P$ contains $y$, then $P$ must also contain $ax = Q_{1} - y^{2}$ and $bx = Q_{2} - y^{2}$.
If $P$ contains $x$, then $P$ contains $\ideal{x, y}$ and by minimality we have $P = \ideal{x, y}$.
If $P$ does not contain $x$, then $P$ contains both $a, b$ which contradicts $\height P = 2$.
This shows that the minimal primes of $\ideal{Q_{1}, Q_{2}}$ are $\ideal{x, y}$ and $\ideal{Q_{1}, a - b}$.
    
The radical of any ideal is the intersection of its minimal primes, and therefore $\radideal{Q_{1}, Q_{2}} = \ideal{x, y} \cap \ideal{Q_{1}, a - b}$.
Suppose $f$ is an arbitrary polynomial in the intersection.
We have $f = g_{1} x + g_{2} y = g_{3} Q_{1} + g_{4} (a - b)$ for polynomials $g_{1}, \dots, g_{4}$.
Since $Q_{1}$ is in the ideal $\ideal{x, y}$, so is $g_{4} (a - b)$.
Since $a, b, x, y$ are linearly independent, $g_{4}$ is in the ideal $\ideal{x, y}$.
This shows that $f$ is in the ideal $\ideal{Q_{1}, x (a - b), y (a - b)}$.
Conversely, every element of this ideal is in the intersection of $\ideal{x, y}$ and $\ideal{Q_{1}, a - b}$.
This shows that $\radideal{Q_{1}, Q_{2}} = \ideal{Q_{1}, x(a - b), y (a - b)} = \ideal{Q_{1}, Q_{2}, y (a - b)}$.
\end{proof}

\section{\sprimes}\label{sec:clean-ideals}

In this section, we define clean vector spaces and prove some useful properties about them, which will be used in our robust Sylvester-Gallai proof.

\subsection{Closeness and robust vector spaces}

For any quadratic form $P\in S$ we would like to measure how far $P$ is from a subalgebra $\bC[V]$ generated by a vector space of forms in $S$. The next definition quantifies the notion of distance from a subalgebra generated by a vector space.

\begin{definition}[Forms close to a vector space]
Given a vector space $V = V_1 + V_2$ where $V_{i} \subseteq S_{i}$, we say that a quadratic form $P$ is \emph{$s$-close} to $V$ if there is a form $Q \in \bC[V]$ such that $\Rank(P - Q) = s$, and for any form $Q' \in \bC[V]$, we have that $\Rank(P - Q') \geq s$.
If a form $P$ is not $r$-close to $V$, for any $r \leq s$, we say that $P$ is \emph{$s$-far} from $V$.

Given a linear form $\ell$, we say $\ell$ is $1$-close to $V$ if $\ell \not \in V_{1}$.
\end{definition}

With the definition above in hand, we are ready to define robust vector spaces. 
These are vector spaces whose quadratic forms are in a sense far from the ideal generated by the linear forms.

\begin{definition}[Robust vector spaces]\label{def:robust-spaces}
     A vector space $V = V_1 + V_2$ where $V_i \subset S_i$ is said to be $r$-robust if, for any nonzero $Q \in V_2$, the following conditions hold: 
     \begin{enumerate}
         \item $Q$ is $(r-1)$-far from $V_1$
         \item if $Q \not\in (V_1)$, then $\Rank(\overline{Q}) \geq r$, where $\overline{Q} \in S/(V_1)$ denotes the image of $Q$ in the quotient ring  $S/(V_1)$.
     \end{enumerate}
     If a homogeneous ideal $I$ has an $r$-robust generating set $V_1 + V_2$, we say that $I$ is an $r$-robust ideal.
\end{definition}

\begin{remark}
Note that if $Q\not\in (V_1)$, then condition $2$ implies condition $1$ for $Q$. Indeed, suppose $\Rank(Q-F)<r$ for some degree $2$ form $F\in \bC[V_1]$. Then $\Rank(\overline{Q})<r$, since $\overline{F}=0$ modulo $(V_1)$. This is a contradiction to condition $2$, therefore we must have  $\Rank(Q - F) \geq r$ for any degree $2$ form $F \in \bC[V_1]$. Therefore, in order to check that a vector space $V=V_1+V_2$ is $r$-robust, it is enough to ensure that condition $1$ holds for all $Q\in V_2\cap(V_1)$ and condition $2$ holds for all $Q\in V_2\setminus (V_1)$.
\end{remark}
\begin{example}\begin{itemize}
    \item [(a)] Let $V_1$ be any subspace of $S_1$, and $V_2=(0)$, then $V=V_1+V_2$ is an $r$-robust vector space for any $r$.
    \item[(b)] Suppose $Q_1,\cdots,Q_m$ are quadratic forms such that $\Rank(Q)\geq r$ for any $Q\in \mathrm{span}(Q_1,\cdots, Q_m)$. Let $V_1=(0)$, $V_2=\mathrm{span}(Q_1,\cdots, Q_m)$. Then $V=V_1+V_2$ is an $r$-robust vector space.
    \item[(c)] Let  $S=\bC[x_1,\cdots,x_r,y_1,\cdots,y_r,u_1,\cdots,u_r]$ and $Q_1=x_1y_1+\cdots+x_ry_r$, $Q_2=u_1y_1+\cdots+u_ry_r$. Let  $V_1=\mathrm{span(x_1,\cdots,x_r)}$, $V_2=\mathrm{span}(Q_1,Q_2)$. Then we can show that $V_1+V_2$ is a $r$-robust vector space. 

Let $Q=aQ_1+bQ_2\in V_2$ for some $a,b\in \bC$. Note that $Q\in (V_1)$ if and only if $b=0$. Suppose, $Q\not\in (V_1)$. Then $\overline{Q}=b\overline{u}_1\overline{y}_1+\cdots+b\overline{u}_r\overline{y}_r$ where $\overline{u}_i,\overline{y}_j$ denote the images of $u_i,y_j$ modulo $(V_1)$. Hence $\Rank(\overline{Q})=r$. Therefore condition $2$, and consequently condition $1$, holds if $Q\in V_2\setminus (V_1)$. 

Now suppose $Q\in V_2\cap (V_1)$. Then $Q=ax_1y_1+\cdots+ax_ry_r$.  Let $F=\sum_ia_{ii}x_i^2+\sum_{i<j}2a_{ij}x_ix_j$, and $A$ be the $r\times r$ matrix with the $(i,j)$-th entry given by $a_{ij}$.
 Let $M$ and $N$ denote the matrices corresponding to the quadratic forms $Q,F$ in $S$. Note that we have $$M=\begin{pmatrix}
0 & \frac{a}{2}I_{r\times r} & 0\\
\frac{a}{2}I_{r\times r}& 0 & 0 \\
0 & 0 & 0
\end{pmatrix}, N= \begin{pmatrix}
A & 0 & 0\\
0& 0 & 0 \\
0 & 0 & 0
\end{pmatrix}$$
where $I_{r\times r}$ denotes the $r\times r$ identity matrix.
Note that the matrix corresponding to $Q-F$ is given by $M-N$. Note that $\Rank(M-N)= \Rank M$, and hence  $\Rank(Q-F)\geq r$.
\end{itemize}

\end{example}

\begin{claim}
  \label{rem:lowrankquadraticinrobust}
  Suppose $V := V_{1} + V_{2}$ is $r$-robust with $V_{i} \subset S_{i}$.
  If $Q \in \bC\bs{V}$ is a quadratic form of rank less than $r$, then $Q \in \bC\bs{V_{1}}$.
  If $P \in \ideal{V}$ is a quadratic form of rank less than $r$, then $P \in \ideal{V_{1}}$.
\end{claim}
\begin{proof}
    Since $Q \in \bC\bs{V}$, and since $Q$ has degree $2$, we have $Q = Q_{1} + Q_{2}$ with $Q_{2} \in V_{2}$ and $Q_{1} \in \bC\bs{V_{1}}$.
    If $Q_{2} \neq 0$, then $Q_{2} - (- Q_{1})$ has rank less than $r$, which contradicts the robustness of $V$.
    Therefore we have $Q_{2} = 0$.
    If $P = P_{1} + P_{2}$ with $P_{1} \in \ideal{V_{1}}$ and $P_{2} \in V_{2}$, then $\overline{P_{2}}$ has rank at most $\Rank{P} < r$ in $S / \ideal{V_{1}}$, and therefore $P_{2} \in \ideal{V_{1}}$.
\end{proof}

\begin{proposition}\label{proposition:2-collapse-robust}
If $V := V_1 + V_2$ is a $5$-robust vector space and $x, y, u, v \in S_1$ such that $xy - uv \in \bC[V]$ then one of the following holds:
\begin{enumerate}
    \item We have $x, y, u, v \in V_1$.
    \item There exists a linear form $\ell \in S_1 \setminus V_1$ such that $x,y,u,v\in \Cspan{V_1,\ell}$.
\end{enumerate}
\end{proposition}
\begin{proof}
Let $P := xy-uv$. Note that $\Rank(P) \leq 2$.
Since $P$ is of degree $2$, we have $P \in \bC\bs{V_{1}}$ by \cref{rem:lowrankquadraticinrobust}.
Since $P\in \bC[V_1]$, we know that $\linold{P}\subset V_1$ by \cite[Fact 2.14]{PS20a}.
Now, \cref{lemma:almost-minimal-collapse-poly-ring} applies and we are done. 
\end{proof}

Another important property of robust vector spaces is that we can get a good control on the vector spaces $\lin{P}$ whenever $P$ is almost in the algebra. 
This is captured by the following proposition, which motivates the subsequent definition of relative space of linear forms.

\begin{proposition}[Quadratics close to robust vector spaces]\label{prop:uniqueness-decomposition-robust-ideal}
    Let $V = V_1 + V_2$ be an $r$-robust vector space and $s < r/2$. 
    If $P$ is $s$-close to $V$, then for any $Q, Q' \in \bC[V]$ such that $\Rank(P - Q) = \Rank(P - Q') = s$, we have that 
    $$\linold{P - Q} + V_1 = \linold{P-Q'} + V_1. $$
    In other words, $(\linold{P-Q}+V_1)/V_1 = (\linold{P-Q'}+V_1)/V_1$ for any two decompositions.
\end{proposition}

\begin{proof}
    Let $R = P-Q$ and $R' = P - Q'$. 
    Thus, we have that $R - R' = Q' - Q \in \bC[V]$ and we have $\Rank(Q' - Q) = \Rank(R - R') \leq \Rank(R) + \Rank(R') \leq 2s < r$.
    Hence, by Remark \ref{rem:lowrankquadraticinrobust}, we have that $Q'-Q \in \bC[V_1]$.
    Now, from $R = R' + (Q' - Q)$ and $Q'-Q \in \bC[V_1]$, we have that $\linold{R} \subseteq \linold{R'} + V_1$, and similarly, we have that $\linold{R'} \subseteq \linold{R} + V_1$.
\end{proof}

The above proposition motivates us to define the relative space of linear forms with respect to a robust algebra.

\begin{definition}[Relative space of linear forms]\label{definition:relative-lin} Let $r\geq 9$ be an integer.
If $V$ is an $r$-robust vector space and $P$ is $s$-close to $V$ for $s < r/2$ we can define 
$$ \linv{V}{P} := 
\begin{cases}
\lin{P} + V_{1}, \text{ if } P \in S_{1} \\
\lin{P-Q} + V_1, \text{ if } P \in S_{2}, s \leq 4 \\
\Cspan{P}, \text{ otherwise}
\end{cases} $$
where $Q \in \bC[V]$ is a form such that $\Rank(P-Q) = s$.
We also define the quotient space
$$ \linvb{V}{P} := 
\begin{cases}
\linv{V}{P} / V_{1}, \text{ if } s \leq 4 \\
{0}, \text{ otherwise}
\end{cases} $$
\end{definition}

\subsection{\sprimes}

We are now ready to define the main object of this section: \sprimes.
\begin{definition}[\sprimes]\label{def:clean-prime} Let $r \geq 17$ be an integer and $\varepsilon \in (0,1)$. 
Let $\cF := \{Q_1, \ldots, Q_m \} \subset S_{\leq 2}$ be a set of forms. 
Let $V=V_1+V_2$ be a vector space with $V_i\subset S_i$. 
We say that $V$ is an $(r, \varepsilon)$-\sprime  over $\cF$ if the following conditions hold:  
\begin{enumerate}
\item $V$ is an $r$-robust vector space
\item For any $U_1 \subset S_1$ such that $\dim(U_1) \leq 8$, there are less than $\varepsilon m$ forms $Q_j \in \cF$ such that $Q_j$ is $s$-close to $V$ for $1 \leq s \leq 4$ and 
$$\dim\left(\linvb{V}{Q_j} \right) > \dim\left(\linvb{V+U_1}{Q_j} \right).$$
\end{enumerate}

\begin{remark}
The vector space $V + U_{1}$ is $r - 8$ robust, since $U_{1}$ consists only of linear forms and has dimension at most $8$.
Since $r - 8 \geq 9$, the space $\dim\br{\linvb{V+U_1}{Q_j}}$ is well defined.
\end{remark}

If $V=V_1+V_2$ is an $(r, \varepsilon)$-\sprime over $\cF$, then we say that the ideal $(V)$ is an $(r, \varepsilon)$-clean ideal over $\cF$, and similarly the algebra $\bC[V]$ is an $(r, \varepsilon)$-clean algebra over $\cF$. 
\end{definition}

\begin{remark}
Condition 2 in \cref{def:clean-prime} simply captures the concept that $Q_j$ becomes closer to being in the algebra $\bC[V, U_1]$. 
Thus, a clean ideal is one where no small dimensional vector space would bring many forms in our set closer to being in the augmented algebra.
\end{remark}

The following lemma proves that one can clean up ideals over a set of quadratic forms without increasing the number of generators by much. 

\begin{lemma}[Cleaning up vector spaces]\label{lemma:ideal-cleanup}
Let $\cF$ be a set of forms, $J = J_1 + J_2$ be a vector space, and $d_i = \dim(J_i)$.
Given $r \in \bZ$ such that $r \geq 17$ and $0 < \varepsilon < 1$, there exists an $(r, \varepsilon)$-\sprime $I = I_1 + I_2$ over $\cF$ such that:
\begin{enumerate}
  \item $\bC[J] \subset \bC[I]$
  \item $I_2 \subset J_2$ and $J_{1} \subset I_{1}$
  \item $\dim(I_1) = \frac{128}{\varepsilon} + d_1 + 4 \cdot d_2 \cdot r$
\end{enumerate}
\end{lemma}

\begin{proof}
Let $\cF = \{Q_1, \ldots, Q_m\}$. Starting with the given vector space $J$, we will iteratively construct a vector space $I$ that satisfies the desired properties. Our proof consists of three steps. First, we will define the iterative process to clean up a vector space. Second, we will show that this iterative process actually terminates after at most $16/\varepsilon$ steps. Finally, we will show that the final output of the iterative process satisfies all the desired properties.

\emph{The iterative processes}. We will devise two processes to clean up the initial vector space $J$, and the final iterative process will be a combination of these two. The first process helps us achieve robustness and the second process helps us clean up the vector space.
In order to obtain the vector space $I$, we will apply both these processes repeatedly.
\begin{itemize}

\item 
Robustness process:
Let $V' = V'_1 + V'_2$ be a vector space with $V'_i \subset S_i$. 
Start with the vector spaces $V_1 := V'_1$, $V_2 := V'_2$ and $V := V_1 + V_2$.
Perform the following two steps till it is not possible to perform either of them:
\begin{itemize}
    \item If there exists a non-zero form $P \in V_2$ such that $\Rank(P - R) < r$ for some $R \in \bC[V_1]$, set $V_1 := V_1 + \linold{P-R}$, and replace $V_{2}$ by a codimension $1$ subspace of $V_{2}$ that does not contain $P$.
    \item If there exists a non-zero form $P \in V_{2} \setminus \ideal{V_{1}}$ such that $P = P_{1} + Q'$ with $P_{1} \in \ideal{V_{1}}$ and $\Rank(Q') < r$, set $V_{1} := V_{1} + \linold{Q'}$.
\end{itemize}
The above two steps correspond to the two conditions in the definition of robustness.
If we cannot perform either of the above two steps, then the vector space $V$ that we end up with is $r-$robust.
We have $V_{2} \subseteq V'_{2}$ and $V'_{1} \subseteq V_{1}$.
It is also easy to see that both steps above maintain the inclusion $\bC\bs{V'} \subseteq \bC\bs{V}$.

We now control how many times the above steps can be performed.
Each time we perform the first step, we decrease $\dim V_{2}$ by $1$, and therefore the first step can be performed a total of $\dim(V_2')$ times.
Each time we perform the second step, we decrease $\dim \overline{V_{2}}$ by $1$, where $\overline{V_{2}}$ is the space $V_{2}$ quotiented by the vector space $\ideal{V_{1}} \cap V_{2}$.
Therefore, the second step can also be performed at most $\dim(V_2')$ times.
Hence, this iterative process always halts, and it does so in at most $2 \dim(V_2')$ steps.
Further, each time we perform either step, we increase the dimension of $V_{1}$ by at most $2r$.

\item
Clean-up process: Let $V' = V'_1 + V'_2$ be a vector space with $V'_i \subset S_i$. 
Start with the vector spaces $V_1 := V'_1$, $V_2 := V'_2$ and $V = V_1 + V_2$.
While there exists $U_{1} \subseteq S_{1}$ such that $U_{1}$ violates the condition that $V$ is $(r, \varepsilon)$-clean (i.e. condition 2 of \cref{def:clean-prime}), update $V_1 := V_1 + U_{1}$.

Note that we have  $V'_1 \subset V_1 \subset S_1$ and $V_2= V'_2$.
Thus, $V' \subset \bC[V]$ by our construction. 
The space $U_{1}$ has dimension at most $8$, and therefore the dimension of $V_{1}$ increases by at most $8$ whenever we apply this procedure.
\end{itemize}

Now, we can design our main iterative process:
\begin{enumerate}
    \item Start with our vector space $W^{(0)} := J$
    \item Apply the robustness process to $W^{(0)}$ to obtain $U^{(0)}$
    \item While $U^{(k)}$ is not $(r, \varepsilon)$-\sprime over $\cF$
    \begin{itemize}
        \item Let $W^{(k+1)}$ be a vector space that we obtain after applying the clean-up process once to $U^{(k)}$
        \item Let $U^{(k+1)}$ be the space we obtain after applying the robustness process to $W^{(k+1)}$
    \end{itemize}
\end{enumerate}

\emph{Termination.} We will now show that the iterative process must terminate after at most $16/\varepsilon$ steps. We define the following potential function for $U := U_{1} + U_{2}$ a $r$-robust vector space:
$$\Phi\br{\cF, U} = \sum_{i=1}^{m} \dim \linvb{U}{Q_{i}} .$$
Note that this potential function is always non-negative and upper bounded by $8 \cdot m$.

Let $\Phi_k := \Phi(\cF, U^{(k)})$ be the potential function that we defined above applied to the robust vector space $U^{(k)}$.
If $U^{(k)}$ violates condition 2 of \cref{def:clean-prime} (clean vector space), there exists $L_1 \subset S_1$ and there are at least $\varepsilon m$ forms $Q_j \in \cF$ such that $Q_j$ is $s$-close to $U^{(k)}$, for $1 \leq s \leq 4$ and 
\begin{align*} 
\dim \linvb{U^{(k)}}{Q_{j}} 
&= \dim\left(\linv{U^{k}}{Q_j}/U^{(k)}_1 \right) \\
&> \dim\left(\linv{U^{(k)} + L_1}{Q_j}/(U^{(k)}_1 + L_1) \right) \\ 
&\geq 
\dim\left(\linv{U^{(k+1)}}{Q_j}/U^{(k+1)}_1 \right) = \dim \linvb{U^{(k+1)}}{Q_{j}}.
\end{align*}
If we denote by $\cB_k$ be the set of such $Q_j$'s, we have $|\cB_k| \geq \varepsilon m$. 

Let $\cN_k$ be the set of forms $Q_j \in \cF$ that are $4$-far from $U^{(k)}$ but are at most $4$-close to $U^{(k+1)}$.
For any $Q_{j} \in \cN_k$, we have that $\dim \linvb{U^{(k)}}{Q_{j}} = 0$ and $\dim  \linvb{U^{(k+1)}}{Q_{j}} \leq 8$.
Moreover, note that $\cN_a \cap \cN_b = \emptyset$ for any $a < b$, since $\bC[U^{(k)}] \subset \bC[U^{(k+1)}]$. Also, $\cN_k\cap\cB_k= \emptyset$.

Then, we have
\begin{align*}
    \Phi_{k} - \Phi_{k+1} 
    &= \sum_{i=1}^m \br{ \dim  \linvb{U^{(k)}}{Q_{i}} - \dim \linvb{U^{(k+1)}}{Q_{i}} } \\
    &\geq \sum_{Q_i \in \cB_k} \br{ \dim  \linvb{U^{(k)}}{Q_{i}} - \dim \linvb{U^{(k+1)}}{Q_{i}} } - \sum_{Q_i \in \cN_k} \dim \br{ \linvb{U^{(k+1)}}{Q_{i}}  }  \\
    &\geq |\cB_k| - 8 |\cN_k| \\
    & \geq \varepsilon m - 8 |\cN_k|,
\end{align*}
where the first inequality holds because for any term not in $\cN_k$, we have that $\dim \linvb{U^{(k)}}{Q_{i}} \geq \dim \linvb{U^{(k+1)}}{Q_{i}}$. 
The second inequality holds because for the elements in $\cB_k$ the dimension drops by at least 1. 
The last inequality holds because we know that $|\cB_k| \geq \varepsilon m$.

Now, if our iterative process runs for $t$ iterations, we have that
\begin{align*}
    \Phi_0 - \Phi_t 
    &= \sum_{k=0}^{t-1} \Phi_k - \Phi_{k+1} \\
    &\geq t \varepsilon m - 8 \cdot \sum_{k=0}^{t-1} |\cN_k| \\
    &\geq t \varepsilon m - 8m 
\end{align*}
where the last inequality holds because the sets $\cN_k$ are pairwise disjoint and contained in $\cF$.
Since $0 \leq \Phi_k \leq 8m$ by the discussion in the first paragraphs of the proof, we have that $t \varepsilon m \leq 16m$, which implies $t \leq 16/\varepsilon$.
In particular, this shows that our process will terminate.

\emph{Desired properties.} Let $I := U^{(t)}$ be the final vector space we obtained at the end of this process. 
By the above, we have that $I$ is $(r, \varepsilon)$-clean over $\cF$, as the process terminated. 
Moreover, note that $\bC[J] \subseteq \bC[I]$, as the inclusion $\bC[U^{(k)}] \subseteq \bC[U^{(k+1)}]$ is preserved throughout the process, and $\bC[J] \subseteq \bC[U^{(0)}]$ by construction.
We also have that $I_2 \subseteq J_2$, since the robustness process preserves $U^{(k+1)}_2 \subseteq U^{(k)}_2$, and the clean-up process preserves $U^{(k+1)}_2 = U^{(k)}_2$, and $U^{(0)} \subset J_2$.
Additionally, we have that $J_1 \subseteq I_1$, as both processes only increase the size of $U^{(k)}$.

Now, all we have left is to bound the dimension of $I_1$. 
For this, note that each run of the clean-up process adds a space of linear forms of dimension at most $8$, and therefore the total contribution from the clean-up process is $128/\varepsilon$.
To see that the total contribution of the robustness process is at most $4r \cdot d_2$, note that the first step of the process strictly decreases $\dim(U^{(k)}_2)$, whereas the second step of the process strictly increases $\dim(U^{(k)}_2 \cap (U^{(k)}_1))$. 
Moreover, if $\dim(U^{(k)}_2 \cap (U^{(k)}_1)) = \dim(U^{(k)}_2)$ from this point on we can only apply the first step of the process.
Finally, note that the clean-up process does not change the degree two part of the vector space.
Hence, we can apply the first step of the process at most $d_2$ times, and the second process at most $d_2$ times throughout the entire iterative process. 
Note that each time we apply the foregoing steps, we increase the dimension of the linear piece by at most $2r$. 
Putting everything together, we have that $\dim(I_2) \leq d_1 + 128/\varepsilon + 4r \cdot d_2$, as we wanted.
\end{proof}

\begin{remark}
Note that in the process above, if $\cF \subset (J)$, then $\cF \subset (I)$.
\end{remark}

The clean up process is similar to a process used in \cite{PS20a} and \cite{PS20b},
where they find an ideal $V$ generated by linear forms, and a subset $\cI$ of quadratics forms such that every quadratic in their configuration is in the span of elements of $\cI$, and elements of $\cF \cap \ideal{V}$.
They then iteratively find quadratic forms in $\cF$ that have rank $2$ in $S / \ideal{V_{1}}$.
If such a form is found, they add the linear forms to $V_{1}$, and remove a quadratic from $\cI$.
This iterative process does not preserve the algebra generated by $V$ and $\cI$.
In other words, a form $Q \in \cF$ that lies in $\bC\bs{V, \cI}$ might not be in the algebra after one step of the iterative process (although $Q$ is still guaranteed to be in the ideal generated by $V, \cI$, and $Q$ satisfies the stronger property that it is spanned by elements of $\cI$ and $\cF \cap \ideal{V}$).
We require that all forms remain in the algebra, and therefore our clean up process is more involved.

\subsection{Univariate forms over clean vector spaces}

We now turn our attention to a special case of forms which are 1-close to a robust vector space.

\begin{definition}[Univariate forms over robust vector spaces]
Let $V := V_1 + V_2$ be an $r$-robust vector space, where $r \geq 3$ and $V_i \subseteq S_i$, for $i \in \{1,2\}$. 
We say that a form $P$ is \emph{univariate over} $V$ if $P$ is $1$-close to $V$ and $\dim\left( \linvb{V}{P}\right) = 1$. 
Moreover, we define $z_P \in S_1/V_1$ to be the linear form such that $\linvb{V}{P} :=  \Cspan{z_P}$.\footnote{More precisely, we can take $z_P$ in the $S$ canonically simply by identifying the quotient polynomial ring with a polynomial ring in less variables. For simplicity, we will abuse notation in this setting.}
\end{definition}

Once again, given a $3$-robust vector space $V$, we will also say that a form $P$ is univariate over the algebra $\bC[V]$ if $P$ is univariate over $V$ as in the definition above.

\begin{remark}
An immediate consequence of the definition above is that if $P$ is a quadratic form univariate over $V$, then $P \in \bC[V, z_P]$.
Moreover, we can write $P = Q + z(\alpha z + v)$ where $Q \in \bC[V]$, $\alpha \in \bC$ and $v \in V_1$.
Thus, $P \in (V) \iff \alpha = 0$. 
\end{remark}

Having univariate forms over a clean vector space is desirable because if our SG configuration only has univariate forms over a clean vector space, one would expect the ``extra variables'' to behave as a LCC configuration.
As mentioned in the proof outline, this is indeed the case, and in the next section we will formalize this intuition about the extra variables.

\section{Proof of the robust Sylvester Gallai theorem}\label{sec:sg-robust}

In this section we prove our main theorem: $\drsg{\delta}{2}$ configurations (\cref{def:deltasg2config}) have constant dimension as a $\bC$-vector space. Throughout this section, we will assume $\varepsilon := \delta / 10$.

\deltasgstatement*

Before we prove the theorem above, we need some definitions.
Let $\cF$ be our $\drsg{\delta}{2}$ configuration, where we write $\cF = \cF_1 \sqcup \cF_2$ such that $\cF_i = \cF \cap S_i$ and let $m := \abs{\cF}$.
We will assume that $m > 1 / \delta^{54}$, since the statement trivially holds if it is not. 

For a form $Q \in \cF_2$, define the sets 
\begin{align*}
\cF_2(Q) &:= \setbuild{Q' \in \cF_2}{\abs{\radideal{Q, Q'} \cap \cF} \geq 3} \\ 
\cF_1(Q) &:= \setbuild{x \in \cF_1 \ }{\ \abs{\radideal{Q, x} \cap \cF} \geq 3}.
\end{align*}
The sets above account for the quadratic and linear Sylvester-Gallai pairs related to $Q$.
We also partition $\cF_2(Q)$ into three sets, $\Fspan[Q], \Flinear[Q]$ and $\Fsquare[Q]$.
These are the sets of forms in $\cF$ that satisfy the span case, linear case, and the square case of \cref{theorem: radical structure} with $Q$.
The formal definitions of the sets are as follows.
\begin{align*}
  \Fspan[Q] &= \setbuild{Q' \in \cF_2(Q) \ }{\ \abs{\ideal{Q, Q'} \cap \cF} \geq 3}\\
  \cF_{\text{square}}(Q) &= \setbuild{Q' \in \cF_2(Q) \setminus \Fspan[Q] \ }{\ \exists \ \ell \in S_{1}, \ell^{2} = \alpha Q + \beta Q'} \\
  \Flinear[Q] &= \setbuild{Q' \in \cF_2(Q) \setminus \br{\Fspan[Q] \cup \Fsquare[Q] } \ }{\ \exists \ x, y \in S_{1}, Q, Q' \in \ideal{x, y}} 
\end{align*}

Since $\cF$ is a $(\delta, 2)$-SG configuration, for each $Q \in \cF_2$ we have 
$$\abs{\Fspan[Q] \sqcup \Flinear[Q] \sqcup \Fsquare[Q] \sqcup \cF_1(Q) } \geq  \delta m.$$

We define sets $\Fspan, \Flinear, \Fsquare$ and $\Fdeg$ that partition $\cF_2$ as follows.

\begin{align*}
  \Fspan &= \setbuild{Q \in \cF_2}{\abs{\Fspan[Q]} \geq \varepsilon m} \\
  \Flinear &= \setbuild{Q \in \cF_2 \setminus \Fspan}{\abs{\Flinear[Q]} \geq \varepsilon m} \\
  \Fdeg &= \setbuild{Q \in \cF_2 \setminus \Fspan \cup \Flinear \ }{ \ \abs{\cF_1(Q)} \geq \varepsilon m} \\
  \Fsquare &= \setbuild{Q \in \cF_2 \setminus \br{\Fspan \cup \Flinear \cup \Fdeg}}{\abs{\Fsquare[Q]} \geq (\delta - 3 \varepsilon) m} 
\end{align*}

By \cref{theorem: radical structure}, every form in $\cF_2$ belongs to one of the above sets.
We will now begin constructing small algebras which control the sets defined above.

\subsection[Controlling F-square]{Controlling $\Fsquare$}

The following lemma from \cite[Corollary~43]{S19} controls the vector space dimension of the set of forms satisfying the square case with two quadratic forms.

\begin{lemma}
  \label{lem:twosquare}
  Let $Q_{1}, Q_{2} \in S_2$ be linearly independent and $F_{1}, \dots, F_{m} \in S_2$ such that
  $$
    F_{i} = Q_{1} + \ell_{i}^{2} = \beta_{i} Q_{2} + b_{i}^{2}
  $$
  where $\ell_{i}, b_{i} \in S_{1}$ and $\beta_{i} \in \bC$.
  Then, there is a vector space $V$ of dimension at most $4$ that contains every $\ell_{i}$ and $b_{i}$.
  In particular, every $F_{i} \in \bC[Q_1, V]$, which implies $\dim \Cspan{Q_1, Q_2, F_1,\ldots, F_m} \leq 7$. We will denote this vector space $V$ as $V(Q_1,Q_2)$.
\end{lemma}

If $P_{1}, P_{2}, \dots, P_{b}$ are the forms in $\Fsquare$,
\cref{lem:twosquare} implies that for every $i, j$ there exists $V(P_{i}, P_{j}) \subset S_1$ with $\dim(V(P_{i}, P_{j})) \leq 4$ such that 
$$\Fsquare[P_{i}] \cap \Fsquare[P_{j}] \subset \bC\bs{\lin{P_i} + V(P_{i}, P_{j})}.$$ 
We use these spaces in order to control $\Fsquare$.
We begin by proving an auxiliary lemma.

\begin{lemma}\label{lem:FSquaretripleintersection}
    Suppose $P_{1}, P_{2}, P_{3} \in \cF_2$ satisfy 
    $$\abs{\Fsquare[P_{1}] \cap \Fsquare[P_{2}] \cap \Fsquare[P_{3}]} \geq 2.$$
    Then $P_{3} \in \bC[\lin{P_{1}}, V(P_{1}, P_{2})]$.
\end{lemma}

\begin{proof}
    Let $R_{1}$ and $R_{2}$ be two distinct forms in the intersection.
    After scaling $R_1,R_2$ appropriately, we have the equations 
    \begin{align*}
        R_{1} &= P_{1} - \ell_{1}^{2} = \alpha_{1} P_{2} - g_{1}^{2} = \beta_{1} P_{3} - h_{1}^{2} \\
        R_{2} &= P_{1} - \ell_{2}^{2} = \alpha_{2} P_{2} - g_{2}^{2} = \beta_{2} P_{3} - h_{2}^{2}.
    \end{align*}
    
    By \cref{lem:twosquare} each of the linear forms $\ell_{1}, \ell_{2}, g_{1}$ and $g_{2}$ is contained in $V(P_{1}, P_{2})$.
    Further, we have $\beta_1 P_{3} = P_{1} - \ell_{1}^{2} + h_{1}^{2}$, whence $P_{3} \in \bC[\lin{P_{1}}, V(P_{1}, P_{2}), h_{1}]$.
    It therefore suffices to show that $h_{1} \in \bC[\lin{P_{1}}, V(P_{1}, P_{2})]$.
    We perform a case analysis based on the coefficients $\beta_{i}$.
    
    Suppose first that $\beta_{1} = \beta_{2} = \beta$.
    We then have $P_{1} - \beta P_{3} = \ell_1^2 - h_1^2 = \ell_2^2 - h_2^2$, which implies 
    $$  (\ell_1 - \ell_2)(\ell_1 + \ell_2) = (h_1 - h_2)(h_1 + h_2), $$
    which implies that $h_1, h_2 \in \Cspan{\ell_1, \ell_2}$, proving the result in this case.
    
    Suppose then that $\beta_1 \neq \beta_2$.
    We have
    $$ (\beta_2 - \beta_1) P_1 + \beta_1 \ell_2^2 - \beta_2 \ell_1^2 = h_1^2 - h_2^2 .$$
    Note that in this case, we must have $\Rank(P_1)\leq 4$. Therefore, $\lin{P_1}=\linold{P_1}$. 
    Now $h_1^2-h_2^2=(h_1+h_2)(h_1-h_2)$ is a minimal representation. 
    By \cref{lem:linold}
    $$\Cspan{h_1, h_2} \subset \lin{P_1} + \Cspan{\ell_1, \ell_2}.$$
    This in particular shows that $h_{1} \in \bC[\lin{P_{1}}, V(P_{1}, P_{2})]$.
\end{proof}

With the lemma above, we can show that $\Fsquare$ is contained in a small subalgebra.

\begin{lemma}\label{lem:controlFLsquare}
There exist vector spaces $V_{1} \subseteq S_{1}$ and $V_{2} \subseteq S_{2}$ of dimensions at most $O(1/\delta^2)$ such that $\Fsquare \in \bC[V_1, V_2]$.
\end{lemma}

\begin{proof}
    Let $\Fsquare := \{F_1, \ldots, F_t\}$.
    Set $\cU = \cF_{2}$ and $S_i = \Fsquare[F_i]$ for $i \in [t]$.
    In this case, we have $|\cU| \leq m$, each $|\Fsquare[F_i]| \geq \varepsilon m$, and therefore we can apply \cref{lemma:double-cover} to these sets with parameter $\varepsilon$ to obtain a collection of subsets of $\Fsquare$ given by $\cH$, $\cH_1, \ldots, \cH_r$, where $r := |\cH| \leq 2/\varepsilon$ and $r_i := |\cH_i| \leq 8/\varepsilon$.
    
    Letting $\cH := \{P_1, \ldots, P_r\}$, and $\cH_i := \{ Q_{i1}, \ldots, Q_{ir_i} \}$, we define our vector space $V$ as
    \begin{align*}
        V := \sum_{i=1}^r \br{ \lin{P_i} + \sum_{j=1}^{r_i} V(P_{i}, Q_{ij}) }
    \end{align*}
    Define $V_{1}$ and $V_{2}$ to be $V \cap S_{1}$ and $V \cap S_{2}$ respectively.
    The dimension of $V$ is $\bigO{1 / \varepsilon^{2}}$, since $\lin{P}, \lin{Q}$ and $V(P_{i}, Q_{ij})$ each have dimension $\leq 8$.
    
    By construction, every $P_{i}$ and $Q_{ij}$ is contained in the algebra $\bC\bs{V}$.
    Suppose $R \in \Fsquare$ is not in $\cH$ or $\cH_{i}$.
    By \cref{lemma:double-cover}, there are $i, j$ such that 
    $$\abs{\Fsquare[R] \cap \Fsquare[P_{i}] \cap \Fsquare[Q_{ij}]} \geq \epsilon^{3} m / 4^{3} \geq 2$$
    Thus, \cref{lem:FSquaretripleintersection} implies that $R \in \bC[\lin{P_{i}}, V(P_{i}, Q_{ij})]$, and therefore in $\bC[V]$.
\end{proof}

\subsection[Controlling F-linear]{Controlling $\Flinear$}

The following lemma, proved in \cite[Claim 4.4]{PS20a}, shows the existence of a small ideal generated by linear forms containing $\Flinear$.
\begin{proposition}\label{proposition:flinear-ps20}
    There exists $V \subset S_1$ such that $\Flinear \subseteq \ideal{V}$ and $\dim V = \bigO{1 / \delta}$.
\end{proposition}

We strengthen the above lemma in two key ways.
First, we show that every form in $\Flinear$ is \emph{univariate} over the algebra of any \sprime $W$ such that $\Flinear \subset (W)$.
If we further assume that every form in $\cF$ is univariate over $\bC[W]$, then $\Flinear \subset \bC[W]$.

\begin{lemma}
\label{lem:linstructure}
If $W = W_{1} + W_{2}$ is any $(17, \varepsilon/2)$ clean vector space such that $\Flinear \subseteq \ideal{W}$ then every form in $\Flinear$ that is not in $\bC\bs{W}$ is of the form $ax + y^{2}$, with $x, y \in W_{1}$.
In particular $Q$ is univariate over $W$ and $\linvb{W}{Q}$ is spanned by $\overline{a}$. 
\end{lemma}
\begin{proof}
Let $Q$ be a form in $\Flinear \setminus \bC\bs{W}$.
Since $W$ is $17$-robust, and since $Q$ has rank $2$, we have $Q \in \ideal{W_{1}}$ by \cref{rem:lowrankquadraticinrobust}.
We also have $\lin{Q} \not\subset W_{1}$, that is $\dim \br{\linvb{Q}{W}} \geq 1$.
If every $P \in \Flinear[Q]$ is such that $\lin{Q} = \lin{P}$, then the clean condition is violated, since $\dim \br{\linvb{W}{P}} \geq 1$ for every such $P$, while $\dim \br{\linvb{Z}{P}} = 0$ where $Z = \lin{Q} + W$.
Hence, there exists $P \in \Flinear[Q]$ such that $\lin{P} \neq \lin{Q}$, which by \cref{proposition:xy-primary}, implies $Q = y^{2} + ax$ and $P = y^{2} + bx$ for independent linear forms $a, b, x, y$.

Consider now the equation $y^{2} + ax = 0$ in the quotient ring $S / \ideal{W_{1}}$.
Suppose $y \in W_{1}$.
In this case, $ax = 0$ in $S / \ideal{W_{1}} \then x \in W_{1}$ or $a \in W_{1}$.
In either case, after potentially swapping the roles of $a$ and $x$ we have the required statement.
Suppose now that $y \not \in W_{1}$.
The equation $y^{2} + ax = 0$ implies that $\gamma y =  a$ in $S_{1} / W_{1}$ for some nonzero constant $\gamma$.
We then have $y^{2} + ax = y(y + \gamma x) = 0$ in $S / \ideal{W_{1}}$.
Since $y \not \in W_{1}$, we must have $y + \gamma x \in W_{1}$.
We can write $Q = (y + \gamma x)^{2} + x (a - 2 \gamma y - \gamma x^{2})$.
After changing $y, x, a$, this is equivalent to the case of $y \in W_{1}$, and the required result holds.
\end{proof}

\begin{claim}\label{claim:flinear-clean-dependence}
Let $U \subset S_1$ be a vector space and $Q = ax + y^2$ where $x, y \in U$, $a \not\in U$.
Suppose $P_{1}, P_{2}$ are two distinct forms in $\Flinear[Q] \cap \bC\bs{U}$, and suppose $R_{1} \in \radideal{Q, P_{1}} \cap \cF_{2}$ and $R_{2} \in \radideal{Q, P_{2}} \cap \cF_{2}$.
Then $R_{1}$ and $R_{2}$ are distinct elements of $\cF_{2}$.
Further, $R_{1}, R_{2} \not \in \bC\bs{U}$, and $R_{1}, R_{2} \in \bC\bs{U + \lin{Q}}$.
\end{claim}

\begin{proof}
    Since $\lin{Q} \neq \lin{P_{i}}$ for $i \in \{1,2\}$, we can use \cref{proposition:xy-primary}, to write $Q = a_i x_i + y_i^{2}$ and $P_{i} = b_i x_i + y_i^{2}$ for independent linear forms $a_i, b_i, x_i, y_i$.
    We now show that after a linear change of coordinates we can make $x_i \in (x)$ and $y_i \in (y)$, for $i \in \{1,2\}$.
    
    Since $P_i \in \bC[U]$, we have $x_i, y_i, b_i \in U$.
    From $Q = xa + y^2 = x_i a_i + y_i^2$, we have $y^2 - y_i^2 \equiv x_i a_i \bmod (x)$.
    If $x_i \not\in (x)$, then by factoriality of $S/(x)$ we have $a_i \in (x, y , y_i) \subset U \then \bL(Q) \subset U$ which is a contradiction.
    Hence, we have that $x_i = \alpha x$ for some $\alpha \in \bC^*$.
    This implies $x(a - \alpha a_i) = y_i^2 - y^2$, and if $a \neq \alpha a_i$, factoriality implies that $y_i \in \Cspan{x, y}$.
    Moreover, we know that $y_i \not\in (x)$ since $Q$ is irreducible, so writing $y_i = \alpha_i x + \beta_i y$ we have $\beta_i \neq 0$ and hence after a change of coordinates we can map $y_i \mapsto y$, as we wanted.
    
    After scaling appropriately, we have $Q = ax + y^{2}$, and $P_{1} = bx + y^{2}$ and $P_{2} = cx + y^{2}$.
    By \cref{proposition:radical-linear}, $\radideal{P_{1}, Q} = (P_{1}, Q, y(a-b))$ and thus $R_{1} \in \radideal{P_{1}, Q} \then R_{1} = \alpha_1 Q + \beta_1 P_{1} + \delta_1 y(a - b)$.
    If $R_{1} \in \bC\bs{U}$ then $\alpha_1 ax + \delta_1 a y \in \bC\bs{U}$, which implies that $a (\alpha_1 x + \delta_1 y) \in \bC\bs{U}$.
    Since $x$ and $y$ are linearly independent, this implies $a \in U$ which is a contradiction, so $R_{1} \not \in \bC\bs{U}$.
    Similarly, we have $R_{2} = \alpha_{2} Q + \beta_{2} P_{2} + \delta_{2} y (a - c)$.
    This proves the last part of the claim.
    
    Suppose $R_{1} = \mu R_{2}$ for some $\mu \in \bC$.
    After rescaling $\alpha_{2}, \beta_{2}, \delta_{2}$ we have
    $$\alpha_1 Q + \beta_1 P_{1} + \delta_1 y(a - b) = \alpha_{2} Q + \beta_{2} P_{2} + \delta_{2} y (a - c).$$
    Plugging in the formulae of $Q, P_{1}, P_{2}$ and rearranging, we get $a \br{\br{\alpha_1 - \alpha_{2}} x + \br{\delta_1 - \delta_{2}}y} \in \bC\bs{U}$.
    Since $a \not \in U$, we must have $\br{\alpha_1 - \alpha_{2}} x + \br{\delta_1 - \delta_{2}}y = 0$, which by the independence of $x$ and $y$ implies $\alpha_1 = \alpha_{2}$ and $\delta_1 = \delta_{2}$.
    The equation therefore implies $\beta_1 P_{1} - \beta_{2} P_{2} = \delta_1 y(b - c)$.
    We also have $\beta_1 P_{1} - \beta_{2} P_{2} = (\beta_1 - \beta_{2})y^{2} + x (\beta_1 b - \beta_{2} c)$.
    Equating the two and rearranging, we get 
    $$x \br{\beta_1 b - \beta_{2} c} = y\br{\delta_1 \br{b - c} - \br{\beta_1 - \beta_{2}}y}.$$
    If $(\beta_1 b - \beta_{2} c) = \br{\delta_1 \br{b - c} - (\beta_1 - \beta_{2})y} = 0$ then $b = (\beta_{2} / \beta_1) c$ and $\beta_{2} \neq \beta_1$ (otherwise $P_{1} = P_{2}$), and $y \in (b)$, contradicting irreducibility of $P_{1}$.
    Otherwise, $y \in \Cspan{b, c}$ and $x \in \Cspan{b, c, y}$ which together imply that $x, y \in \Cspan{b, c}$.
    This implies that $P_{1}$ is bivariate, contradicting its irreducibility.
    Therefore, $R_{1} \neq \mu R_{2}$ for any $\mu$, that is, $R_{1}$ and $R_{2}$ are distinct elements of $\cF_{2}$.
\end{proof}

\begin{lemma}\label{lem:controlFLinear}
Let $W = W_{1} + W_{2}$ be a $(17, \varepsilon / 3)$ clean vector space such that $\Flinear \subseteq \ideal{W}$ and such that every form in $\cF$ is either in $\bC\bs{W}$ or univariate over $W$.
Then $\Flinear \subset \bC\bs{W_{1}}$.
\end{lemma}

\begin{proof}
Assume that there is $Q \in \Flinear \setminus \bC\bs{W_{1}}$. By \cref{lem:linstructure}, we know that $Q=ax+y^2$, with $x,y\in W_1$ and $a\not\in W_1$.
Let $Z := \lin{Q} + W$.
Every form in $\Flinear[Q]$ has rank $2$, and any such form that is in $\bC\bs{W}$ is also in $\bC\bs{W_{1}}$. Also, if  a form in $\Flinear$ has rank $2$ and it is univariate over $W$
 then it is univariate over $W_1$.
 
Let $\{ P_{1}, \dots, P_{b} \} := \Flinear[Q] \cap \bC\bs{W_{1}}$. 
If $R_{i} \in \radideal{Q, P_{i}}$, then by \cref{proposition:radical-linear} $R_i \in \cF_2$ and \cref{claim:flinear-clean-dependence} shows that $R_{1}, \dots, R_{b}$ are distinct.
It also shows that $R_{i} \not \in \bC\bs{W_{1}}$, that is, $\dim \br{\linvb{W}{R_i}} > 0$ and that $R_{i} \in \bC\bs{Z}$, that is, $\dim \br{\linvb{Z}{R_i}} = 0$.
By the clean condition, we have $b < \varepsilon m / 3$.

Hence, we have $|\Flinear[Q] \setminus \bC\bs{W_{1}}| \geq 2 \varepsilon m/3$.
By the clean condition there are at most $\varepsilon m / 3$ forms $P' \in \cF$ such that $\linvb{W}{Q} = \linvb{W}{P'}$, since for each such form we have $\dim\br{\linvb{W}{P'}} = 1$ while $\dim\br{\linvb{Z}{P'} } = 0$.
We can therefore find a form $P \in \Flinear[Q] \setminus \bC\bs{W_{1}}$ and such that $\linvb{W}{Q} \neq \linvb{W}{P}$.
Note that this latter condition also implies $\lin{Q} \neq \lin{P}$. Also we must have $\lin{Q} \cap \lin{P} \subset \Cspan{x,y}$. Indeed, suppose $u\in \lin{Q} \cap \lin{P}$. As $\lin{Q}=\Cspan{a,x,y}$, we must have $\overline{a}=\lambda \overline{u} \in \linvb{W}{P} $ for some non-zero $\lambda$. This is a contradiction since $\linvb{W}{Q} \neq \linvb{W}{P}$ and both are one dimensional.

Since $\lin{P} \cap \lin{Q} \subset \Cspan{x, y}$ and $\lin{Q} \neq \lin{P}$, after a linear change of variables we can write $Q = y^{2} + ax$ and $P = y^{2} + bx$ for independent linear forms $a, b, x, y$ by \cref{proposition:xy-primary}.

Let $R \in \radideal{Q, P} \cap \cF_2$. 
As $R \in \cF_2$, we have that $R$ is irreducible and univariate over $W_{1}$ or $R \in \bC[W_1]$.
By \cref{proposition:radical-linear} we have 
$$R = \alpha Q + \beta P + \gamma y (a - b) = (\alpha + \beta) y^{2} + x (\alpha a + \beta b) + \gamma y(a - b).$$
Let $R' := x (\alpha a + \beta b) + \gamma y(a - b)$, so $R - R' = (\alpha + \beta )y^{2} \in \bC\bs{W_{1}}$.
The form $R$ is either univariate over $W_{1}$ or in $\bC\bs{W_{1}}$ by assumption, and therefore so is $R'$.
Since $a, b, x, y$ are independent, if $\alpha \neq - \beta$ then $x, y, \alpha a + \beta b, a - b$ are independent, and $R'$ is irreducible with $\lin{R'} = \Cspan{x, y, a, b}$.
This would imply that $\dim \br{\linvb{W}{R'}} = 2$, with basis $a, b$ contradicting the fact that $R'$ is univariate over $W_{1}$ or in $\bC\bs{W_{1}}$.
Therefore, we have $\alpha = -\beta$ and we get $R = (\alpha x + \gamma y) (a - b)$.
This contradicts irreducibility of $R$ and completes the proof.
\end{proof}

\subsection[Controlling F-deg]{Controlling $\Fdeg$}

Now we control $\Fdeg$. 
The proof is similar to the proof in \cref{lem:controlFLinear}, where we construct a double cover of the dependencies to be able to find a small vector space of forms to add to our algebra. 
Before we execute the above plan, we need to clean-up the quadratics in $\Fdeg$, which is the purpose of the following two propositions.

\begin{proposition}\label{proposition:same-lin}
Let $\cN \subseteq \Fdeg$ be the set of quadratics $Q$ such that there are at least $\varepsilon^2 m/100$ other forms $P \in \cF$ such that $\lin{Q} = \lin{P}$.
There exists $W \subset S_1$ with $\dim W = O(1/\delta^2)$ such that $\cN \subset \bC[W]$.
\end{proposition}

\begin{proof}
For each $Q \in \cN$, we must have $\Rank(Q) \leq 4$, otherwise $\lin{Q} = \Cspan{Q}$ and by the SG condition there is no other form $P \in \cF$ such that $\lin{P} = \lin{Q}$, as that would imply $P \in \Cspan{Q}$.
Let $\cN(Q)$ be the set of forms $P \in \cF$ such that $\lin{P} = \lin{Q}$.

Start with $\cA = \emptyset$, and $W = \bc{0}$.
Iterate over the elements of $\cN$, and whenever $Q \in \cN$ is not in the algebra $\bC\bs{W}$, we add $Q$ to $\cA$, and update $W \leftarrow W + \lin{Q}$.
At the end of the process, we have $\cN \subseteq \bC\bs{W}$ by construction.

Since $\dim(W) \leq 8 \cdot |\cA|$, it is enough for us to bound the size of $\cA$.
Suppose the forms of $\cA$ are $Q_{1}, \dots, Q_{t}$ in the order they were added.
Since we add $Q_{i}$ to $\cA$, it must be that $\cN\br{Q_{i}} \cap \cN\br{Q_{j}} = \emptyset$ for every $j < i$.
We therefore have 
$$ \abs{\cup_{k \leq i} \cN\br{Q_{k}}} -  \abs{\cup_{k < i} \cN\br{Q_{k}}} \geq \varepsilon^2 m/100.$$
Since $\abs{\cup_{k \leq i} \cN\br{Q_{k}}} \leq |\cF| \leq m$, these equations together with $\varepsilon = \delta/10$ imply $t = \bigO{1 / \delta^2}$.
\end{proof}

\begin{proposition}\label{proposition:non-radical-fdeg}
Let $\cM \subseteq \Fdeg$ be the set of quadratics $Q$ such that  $(Q, x)$ is not radical for at least $\varepsilon m/3$ linear forms $x \in \cF_1(Q)$.
There is $U \subset S_1$ such that $\dim(U) = O(1/\delta)$ and $\cM \subset \bC[U]$.
\end{proposition}

\begin{proof}
By \cref{proposition:irreducible-radical-lin}, if two forms $Q_1, Q_2 \in \cF_2$ do not form a radical ideal with at least $3$ linear forms from $\cF_1$, then $\Rank(Q_1) = \Rank(Q_2) = 2$ and $\lin{Q_1} = \lin{Q_2}$.
Moreover, the same proposition implies that each form in $\cM$ has rank 2.
In particular, $\lin{Q} \subset S_1$ for each $Q \in \cM$.

Denote $\cM := \{P_1, \ldots, P_r\}$ and for each $i \in [r]$, let 
$$\cH_i := \{ x \in \cF_1(P_i) \ \mid \ (P_i, x) \text{ not radical} \}.$$
By definition of $\cM$, $|\cH_i| \geq \varepsilon m/3$ for each $i \in [r]$.
Since $|\cF_1| \leq m$, \cref{proposition:cover-basic} applies. 
Let $Q_1, \ldots, Q_t$ be the forms from $\cM$ which we obtain from \cref{proposition:cover-basic}, where $t \leq 6/\varepsilon$.

For any $Q \in \cM$, there exists $i \in [t]$ such that $|\cF_1(Q) \cap \cF_1(Q_i)| \geq \varepsilon^2 m/36 > 2$.
In particular, the above together with \cref{proposition:irreducible-radical-lin} implies that $\lin{Q} = \lin{Q_i}$.

Therefore, if we take $U := \sum_{i=1}^t \lin{Q_i}$, we have that $\dim(U) \leq 8 \cdot t \leq 48/\varepsilon = O(1/\delta)$ and $\cM \subset \bC[U]$ by the previous paragraph.
\end{proof}

By the previous two propositions, we are left with the forms $P \in \Fdeg$ such that:
\begin{enumerate}
    \item $(P, x)$ is radical with at least $2 \varepsilon m/3$ of the linear forms $x \in \cF_1(P)$
    \item $\lin{P} = \lin{Q}$ for at most $\varepsilon^2 m/100$ forms $Q \in \cF_2$.
\end{enumerate}
Let $\cR$ denote the set of such forms. 
The next remark shows that for any $P \in \Fdeg$, we can find many distinct quadratic dependencies corresponding to the linear forms satisfying item 1 above. \\

\begin{remark}[Quadratic dependencies]\label{proposition:properties-fdeg-p}
For each form $P \in \cR$ and $x_i \in \cF_1(P)$ such that $(P, x_i)$ is radical, the SG condition implies that there is $Q_i \in \cF_2$ such that $Q_i \in (P, x_i)$.
By \cref{proposition:factor-in-pencil} there are at most $6$ linear forms $x_j$ such that $Q_i \in (P, x_j)$.
Hence, in the set 
$$\tau(P) := \{ (Q_i, x_i) \in \cF_2 \times \cF_1(P) \mid Q_i \in (P, x_i) \text{ and } (P, x_i) \text{ radical}  \}$$
we have that a quadratic form $Q \in \cF_2$ can appear in at most $6$ tuples, and each $x \in \cF_1(P)$ has a $Q \in \cF_2$ such that $(Q, x) \in \tau(P)$.
Thus, we know that there is a subset 
$$\Gdeg[P] := \{(Q_1, x_1), \ldots, (Q_t, x_t)  \} \subset \tau(P)$$ 
with $t \geq \varepsilon m/10$ such that for every $i \in [t]$ we have $\bL(Q_i) \neq \bL(P)$ and for $i \neq j$ we have $Q_i \not\in (Q_j)$ and $x_i \not\in (x_j)$.
With this at hand, we can define the set
$$\Fdeg[P] := \{ Q \mid \exists x \in \cF_1(P) \text{ s.t. } (Q, x) \in \Gdeg[P] \}. $$
By the above, we have that $|\Fdeg[P]| \geq \varepsilon m/10$ for any $P \in \cR$.
\end{remark}

\vspace{8pt}

Now, to control the forms in $\cR$, we will construct a small dimensional $(17, \varepsilon^3 /100)$-clean vector space $V$ such that every form $P \in \cR$ is either in $\bC[V]$ or $P$ is univariate over $V$.
Moreover, if the latter happens, we can prove that all but $\varepsilon^3 m/100$ of the SG dependencies in $\Fdeg[P]$ are also $1$-close to $V$.
Before we construct this vector space $V$, we need the following proposition.

\begin{lemma}\label{lem:aux-controlFdeg}
Let $P_1, P_2, P_3$ be quadratics in $\cR$ such that $\lin{P_i} \neq \lin{P_3}$ for $i \in \{1, 2\}$ and let $V = V_1 + V_2$ be a $(17, \gamma)$-\sprime with respect to $\cF$ such that $P_1, P_2 \in \bC[V]$.

\begin{enumerate}
    \item  If $P_3 \not\in \bC[V]$, then $|\Fdeg[P_3] \cap \bC[V]| \leq \gamma m$. 
    \item  Suppose
$$ \abs{\Fdeg[P_1] \cap \Fdeg[P_2] \cap \Fdeg[P_3]} = B > 3 \gamma m + 2.$$
Then $P_3$ is univariate over $\bC[V]$.
Moreover, if $P_3 \not\in \bC[V]$, there exist $\alpha \in \bC^*$, $z \in S_1 \setminus V_1$, $a \in V_1$ such that one of the following holds: 
\begin{enumerate}
\item $P_3 = \alpha P_1 + az$, and there exist $t \geq B - 3 \gamma m - 2$ forms $R_1, \ldots, R_t \in \Fdeg[P_1] \cap \Fdeg[P_2] \cap \Fdeg[P_3]$, such that 
\begin{align*}
    R_i &= \alpha P_1 - x_i a = P_3 - a_i a
\end{align*} 
where $x_i \in \cF_1(P_1) \setminus V_1$ and $a_i \in \cF_1(P_3) \setminus V_1$.
\item $P_3 = \alpha P_1 + z(z+a)$, and there exist $t \geq B - 3 \gamma m - 2$ forms $R_1, \ldots, R_t \in \Fdeg[P_1] \cap \Fdeg[P_2] \cap \Fdeg[P_3]$, such that
\begin{align*}
    R_i &= \alpha P_1 - x_i (x_i + \alpha_i a) = P_3 - a_i b_i
\end{align*} 
where $\alpha_i \in \bC$, $x_i \in \cF_1(P_1) \setminus V_1$ and $a_i \in \cF_1(P_3) \setminus V_1$.
\end{enumerate}
Furthermore, in both cases above, we have that $\overline{z}, \overline{x_i}, \overline{a_i} \in S_1/V_1$ are pairwise linearly independent and $\dim \Cspan{\overline{z}, \overline{x_i}, \overline{a_i}} = 2$. 
\end{enumerate}

\end{lemma}

\begin{proof}
We begin by proving item 1.

\paragraph{Proof of item 1:} Suppose $|\Fdeg[P_3] \cap \bC[V]| > \gamma m$. Let $R_i\in |\Fdeg[P_3] \cap \bC[V]| $ for $i \in [\gamma m+1]$. In this case, we will show that $P_3 \in \bC[V]$.  After proper scalar multiples, we may assume $R_i=
P_3-a_ib_i$ where $a_i\in \cF_1(P_3)$ and $b_i\in S_1$.
Assume for the sake of contradiction that $P_3 \not\in \bC[V]$.
Suppose there is $i \in [\gamma m+1]$ such that $a_i \in V_1$. Without loss of generality we can assume that $i = 1$. 
If $P_3 \not\in \bC[V]$ then $b_1 \not\in V_1$.
However, we have $a_1 b_1 - a_ib_i = R_i - R_1 \in \bC[V_1]$, which implies that $b_i \in (a_1)$ and $a_i \in \Cspan{V_1, b_1} \setminus V_1$, since $a_i \not\in (a_1)$. 
This happens for all $2 \leq i \leq \gamma m+1$.
We therefore have $\dim \linvb{V}{P_{3}} > \dim \linvb{V + \bc{b_1}}{P_{3}}$ and $\dim \linvb{V}{a_{i}} > \dim \linvb{V + \bc{b_1}}{a_{i}}$ for every $2 \leq i \leq \gamma m$, which contradicts the clean condition.

We are left with the case where $a_i \not\in V_1$ for each $i \in [\gamma m+1]$.
In this case, $a_1 b_1 - a_ib_i = R_i - R_1 \in \bC[V_1]$ and \cref{proposition:2-collapse-robust} imply that there is $z_i \in S_1 \setminus V_1$ such that $a_1, b_1, a_i, b_i \in \Cspan{V_1, z_i}$.
Since $a_1 \not\in V_1$, we must have that $\Cspan{V_1, z_i} = \Cspan{V_1,  a_1 }$, which implies that $a_i \in \Cspan{V_1, a_1}$ for all $i \in [\gamma m]$.
Equivalently, we have $\dim \linvb{V}{a_{i}} = 1$ and $\dim \linvb{V + \bc{a_{1}}}{a_{i}} = 0$ or all such $i$, which violates the clean condition.
Therefore, this case cannot happen.

\paragraph{Proof of item 2:} We can assume that $P_3 \not\in \bC[V]$, otherwise we are done.
Let $R_i$ be elements in the intersection $\Fdeg[P_1] \cap \Fdeg[P_2] \cap \Fdeg[P_3]$, where $i \in [B]$. 
After proper scalar multiples, we have:
\begin{equation}\label{eq: fdeg triple dependencies}
R_i = \alpha_i P_1 - x_i y_i = \beta_i P_2 - u_i v_i = P_3 - a_i b_i 
\end{equation}
where $\alpha_i, \beta_i \in \bC^*$, $x_i \in \cF_1(P_1)$, $u_i \in \cF_1(P_2)$ and $a_i \in \cF_1(P_3)$.
Moreover, by \cref{proposition:properties-fdeg-p} we know that for any $i \neq j$ we must have $x_i \not\in (x_j)$, $u_i \not\in (u_j)$ and $a_i \not\in (a_j)$.

By the first item and the fact that $P_3 \not\in \bC[V]$, at least $B_{1} := B - \gamma m$ forms $R_i$ not in the algebra. (Say $R_i \not\in \bC[V]$ for $i \in [B_{1}]$.)
Recall also that $\lin{P_3} \neq \lin{P_1}$.

Since $P_3 - \alpha_i P_1 = x_i y_i - a_i b_i$, we know that $\Rank(P_3 - \alpha_i P_1) \leq 2$.
For any $i \in [B_{1}]$ such that $\Rank(P_3 - \alpha_i P_1) = 2$, we have that $P_3 - \alpha_i P_1 = x_i y_i - a_i b_i$ is a minimal representation of $P_3 - \alpha_i P_1$, and therefore $\Cspan{x_i, a_i, y_i, b_i} = \lin{P_3 - \alpha_i P_1}$. 
Let $\cL \subseteq [B_1]$ be the indices $i$ such that $\Rank(P_3 - \alpha_i P_1) = 2$. 

If there were $i \neq j \in \cL$ such that $\alpha_i \neq \alpha_j$, w.l.o.g. we can say $\alpha_1 \neq \alpha_2$, then $\lin{P_3}, \lin{P_1} \subset W := \Cspan{x_1, x_2, y_1, y_2, a_1, a_2, b_1, b_2}$.
For every other $k \in \cL$, as $\Rank(x_{k} y_{k} - a_{k} b_{k}) = 2$, then $x_{i}, y_{i}, a_{i}, b_{i} \in W$.
This would imply that $R_{i} \in \bC\bs{V + W}$, or equivalently that $\dim \linvb{V}{R_{i}} > \dim \linvb{V + W}{R_{i}} = 0$
The above cannot happen for more than $\gamma m$ indices, as it would contradict the clean condition, since $\dim(W) \leq 8$.

If $\alpha_i = \alpha_j$ for all $i \neq j \in \cL$, we obtain the same conclusion as above, but now with $W := \Cspan{x_1, y_1, a_1, b_1}$.
Therefore, we know that $|\cL| \leq \gamma m$.

We have just proved that there must be at least $B_2 := B_{1} - \gamma m$ indices with $P_3 - \alpha_i P_1 = x_{i} y_{i} - a_{i} b_{i}$ being rank $1$.
Moreover, by \cref{remark:factor-in pencil} and $\lin{P_1} \neq \lin{P_3}$  there is a unique rank 1 element in the pencil $P_3 - \lambda P_1$.
In particular, this implies $\alpha_i = \alpha$ for all $i \in [B_2]$.

After rescaling $P_1$ and reordering the indices, we can assume that $P_3 - P_1$ is of rank 1, and that $x_{i} y_{i} - a_{i} b_{i}$ has rank $1$ for every $i \in \bs{B_{2}}$.
Hence, \cref{eq: fdeg triple dependencies} becomes:
\begin{equation}\label{eq: fdeg triple dependencies2}
R_i = P_1 - x_i y_i = \beta_i P_2 - u_i v_i = P_3 - a_i b_i 
\end{equation}
We now have two cases to analyze:

\paragraph{Case (a):} $P_3 \in (V)$.  

In this case, we can write $P_3 = P_1 + za$, where $a \in V_1$ and $z \in S_1 \setminus V_1$. Since $x_i y_i - u_i v_i =  P_1 - \beta_i P_2 \in \bC[V_1]$, and by our assumption we know that $x_iy_i \not\in \bC[V]$, \cref{proposition:2-collapse-robust} implies that there is $z_i \in S_1 \setminus V_1$ such that $x_i, y_i, u_i, v_i \in \Cspan{V_1, z_i}$.
By the clean condition on $V$, we have that $z_i \in \Kspan{V_1, z}$ for at most $\gamma m$ indices.
Thus, if $B_{3} := B_{2} - \gamma m$, we can assume that $z_i \not\in \Kspan{V_1, z}$, for $i \in [B_3]$.  

As $P_1 - x_i y_i = P_3 - a_i b_i$, we have that $a_i b_i - x_i y_i = P_3 - P_1 = za$ for all $i \in [B_{2}]$.
Hence, $a_ib_i \equiv x_iy_i \bmod (a)$. 
If $x_iy_i \not\equiv 0 \bmod (a)$, by the factoriality of $S/(a)$ we have that $a_ib_i = (x_i + \beta_i a)(y_i + \gamma_i a)$, which implies that
$$ za = a_ib_i - x_iy_i = a(\gamma_i x_i + \beta_i y_i + \beta_i \gamma_i a), $$
which in turn implies that $z = \gamma_i x_i + \beta_i y_i + \beta_i \gamma_i a$.
Since $a_i b_i \neq x_i y_i$, we have that one of $\beta_i, \gamma_i$ is nonzero, which implies that $z \in V_1 + \Cspan{x_i, y_i} = V_1 + \Cspan{z_i}$, which implies that $z_i \in V_1 + \Cspan{z}$, contradicting our choice of $z_i$.
This contradiction implies $x_i y_i \in (a)$.
Note that for every such $i$, we also have $a_{i} b_{i} = za - x_{i} y_{i} \in \ideal{a}$.

Since the forms $x_{j}, a_{k}$ for all $j, k$ are pairwise linearly independent, there can be at most one index with $x_{i} \in \ideal{a}$ or $a_{i} \in \ideal{a}$.
Therefore, for $i \in [B_{3} - 2]$, we have $y_{i}, b_{i} \in \ideal{a}$.
In particular, this implies that $\Kspan{V_1, x_i} = \Kspan{V_1, z_i}$ for $i \in [B_3-2]$, and therefore $x_i \not\in \Kspan{V_1, z}$.
After rescaling the $x_{i}$ and $a_{i}$, this implies
$$ R_{i} = P_{1} - x_{i} a = P_{3} - a_{i} a$$
for each such index, completing the first part of the proof in this case.

For each index satisfying the above, the condition $za = a_{i} b_{i} - x_{i} y_{i}$ now implies that $z = \alpha_{i} a_{i} - \beta x_{i}$ for nonzero $\alpha_{i}, \beta_{i}$.
Since $x_{i} \not \in \Kspan{V_{1}, z}$ by our choice of indices, we also have that $a \not\in V_1$.
If $\overline{z}, \overline{a_{i}}, \overline{x_{i}}$ are not pairwise linearly independent, then $z_{i} \in V_{1} + \Cspan{z}$, contradicting the choice of indices.
This completes the proof of the moreover part in this case.

\paragraph{Case (b):} $P_3 \not\in (V)$.

In this case, we have $P_{3} - P_{1} = pq$ with neither $p$ nor $q$ in $V_1$.
Let $W$ be the space spanned by $p, q$ in $S_{1} / V_{1}$.
Also, let $z_{i} \in S_{1} \setminus V_{1}$ be such that $x_{i}, y_{i}, u_{i}, v_{i} \in \Cspan{V_{1}, z_{i}}$ and let $\mu_{i}, \nu_{i} \in \bC$ be such that $x_{i} = \mu_{i} z_{i}$ and $y_{i} = \nu_{i} z_{i}$ in $S_{1} / V_{1}$.
Since $R_{i} \not \in \bC\bs{V}$, at most one of $\mu_{i}, \nu_{i}$ can be zero for each $i$.
We have equation $pq = P_3 - P_1 = \nu_{i} \mu_{i} z_{i}^{2} - a_{i} b_{i}$ in the ring $S_{1} / \ideal{V_{1}}$.

Let $\cL_1 \subset [B_2]$ be the set of indices such that $z_i \in V_1 + W$.
This implies $R_i \in \bC[V, W]$ for $i \in \cL_1$. By the clean condition, we have that $|\cL_1| \leq \gamma m$.
Let $\cL_2 \subset [B_2] \setminus \cL_1$ be the set of indices such that $a_i$ or $b_i \in V_1 + W$.
In this case, we have that $pq - a_ib_i = \mu_i \nu_i z_i^2$, and since $z_i \not\in V_1 + W$, we have that $\mu_i \nu_i = 0$ and thus both $a_i, b_i \in V_1 + W$, hence $R_i \in \bC[V, W]$.
However, this implies $R_i - P_1 = x_i y_i \in \bC[V, W] \then z_i \in V_1 + W$, which is a contradiction, so $\cL_2 = \emptyset$.

Without loss of generality, we can assume that $\bs{B_{3}} = [B_2] \setminus \cL_{1}$.
By definition, $i \in [B_3] \then z_i \not\in V_1 + W$ and $a_i b_i \not\in (V_1 + W)$.
Since $x_i y_i = pq - a_i b_i \equiv -a_ib_i \not\equiv 0 \bmod (V_1 + W)$, this also implies that both $x_i, y_i \not\in (V_1 + W)$.
Therefore, for every $i \in \bs{B_{3}}$ we have $x_{i}, y_{i}, a_{i}, b_{i} \not \in V_{1} + W$.

We now show $\dim W = 1$.
Assume towards a contradiction that $\dim W = 2$.
For each $i \in \bs{B_{3}}$, let $W_{i} := \Cspan{a_i, b_i, z_i}$ over $S_{1} / V_{1}$.
From $x_{i} y_{i} - a_{i} b_{i} = pq$, we have $p, q \in W_{i} \then W \subset W_i$.
Since $z_{i}, a_{i} \not\in W$, this inclusion is strict, and $\dim W_{i} = 3$.

For every $i, j \in \bs{B_{3}}$ we have $a_{i}b_{i} - a_{j} b_{j} = x_{i} y_{i} - x_{j} y_{j} \in \bC\bs{V_{1} + \Cspan{z_{i}, z_{j}} }$.
Hence, by \cref{proposition:2-collapse-robust} there is $\ell_{ij} \in S_1$ such that $a_{i}, b_{i}, a_{j}, b_{j} \in \Cspan{z_{i}, z_{j}, \ell_{ij}} + V_1$.
Therefore, $\Cspan{a_{i}, b_{i}, z_{i}} + V_1 = \Cspan{z_{i}, z_{j}, \ell_{ij}} + V_1 = \Cspan{a_{j}, b_{j}, z_{j}} +V_1 \then W_{i} = W_{j}$.
This implies that $R_{i} \in \bC\bs{V + W_{1}}$ for every $i \in [B_3]$, or equivalently that $\dim \linvb{V}{R_{i}} > \dim \linvb{V + W_1}{R_{i}}$. This contradicts cleanliness, since $B_{3} \geq \gamma m$.

Therefore, it must be that $\dim(W) = 1$, or equivalently (after rescaling) that $q = p + a$ for some $a \in V_{1}$.
Setting $z = p$, we get that $W = \Cspan{z}$ and the first part of case $1$ of the lemma holds, that is $P_{3} = P_{1} + z(z + a)$.
For each $i \in [B_3]$, we also know that $y_{i}, x_{i} \not \in V_{1}$, and therefore, we have $y_{i} = \alpha_i x_{i} + f_{i}$ for some $\alpha_i \in \bC^*$ and $f_{i} \in V_{1}$ (since $x_{i}, y_{i}$ are spanned by $z_{i}$ in $S_{1} / V_{1}$).
Since $x_iy_i - a_ib_i = z(z+a)$, by \cref{proposition:2-collapse-robust} and $x_i \not\in V_1 + W$, we have $x_{i}, y_{i}, a_{i}, b_{i} \in  \Cspan{z, a, x_i}$.
Since $f_{i} = y_i - \alpha_i x_i \in \Cspan{z, a, x_{i}}$ and $x_{i} \not\in V_1 + W$, it must be that $f_{i} \in  \Cspan{a}$.

From $z(z+a) = x_{i}(x_{i} + f_{i}) - a_{i} b_{i}$, we have $z^2 - x_i^2 \equiv a_i b_i \bmod (V_1) \then a_{i} = \beta_i (x_{i} \pm z)$ and $b_{i} = \beta_i^{-1}(x_{i} \mp z)$ over $S / (V_1)$.
Thus, $\overline{a_i}, \overline{z}$ and $\overline{x_i}$ form a linear dependence.
Combined with the fact that $\overline{a_{i}}, \overline{x_{i}} \not\in (\overline{z})$, we also deduce that $\overline{a_{i}}$ and $\overline{x_{i}}$ are linearly independent over $S/(V_1)$.
\end{proof}

\begin{lemma}[Controlling $\Fdeg$]\label{lem:controlFdeg}
There exists a $(17, (\varepsilon/100)^{3})$-clean vector space $V := V_1 + V_2$, where $V_i \subset S_i$ of dimension $\dim(V) = O(1/\delta^3)$ such that each form in $\Fdeg$ is either in $\bC[V]$ or univariate over $V$. 
Moreover:
\begin{enumerate}
    \item When $P \in \Fdeg \setminus \bC\bs{V}$, if $z_{P}$ spans $\linvb{V}{P}$ then there is $t \geq (\varepsilon/100)^{3} m$ and distinct linear forms $x_{1}, \dots, x_{t}\in \cF_1$, and distinct linear forms $a_{1}, \dots, a_{t}\in \cF_1$ such that for every $i \in [t]$, we have $x_{i}, a_{i} \not \in V_{1}$, $z_{P}, \overline{x_{i}}, \overline{a_{i}}$ are pairwise linearly independent in $S_{1} / V_{1}$, and $z_{P} \in \Cspan{\overline{x_{i}}, \overline{a_{i}}}$.
    \item If $W$ is $(17, (\varepsilon/100)^{3})$-clean such that $\bC\bs{V} \subseteq \bC\bs{W}$, then the above condition holds with $W$ too.
\end{enumerate}
\end{lemma}

\begin{proof}
Partitioning $\Fdeg = \cM \sqcup \cN \sqcup \cR$ as defined in this subsection, \cref{proposition:non-radical-fdeg} implies that there exists a vector space $U \subset S_1$ such that $\dim(U) = O(1/\delta)$ and $\cM \subset \bC[U]$.
Similarly, \cref{proposition:same-lin} implies that there is a vector space $W \subset S_1$ such that $\dim(W) = O(1/\delta^2)$ and $\cN \subset \bC[W]$.
Thus, we are only left with the forms in $\cR$.

\cref{proposition:properties-fdeg-p} implies that $|\Fdeg[P]| \geq \varepsilon m/10$ for each $P \in \cR$.
Thus, applying our double cover lemma \cref{lemma:double-cover} with the sets $\Fdeg[P_i]$, for $P_i \in \cR$, and parameter $\nu = \varepsilon/10$, we obtain a set of forms $\cG \subset \cR$ with $|\cG| \leq (200/\varepsilon)^2$ such that for any form $P \in \cR$, we have that either $P \in \cG$ or there exist $P_1, P_2 \in \cG$ such that 
$$ \abs{\Fdeg[P] \cap \Fdeg[P_1] \cap \Fdeg[P_2]} \geq (\nu/4)^3 m 
= \varepsilon^{3} m/40^3 .  $$
Now, letting $V' := U + W + \Cspan{\cG}$, and applying the cleanup lemma \cref{lemma:ideal-cleanup} with robustness parameter $r = 17$ and clean parameter $(\varepsilon/100)^{3}$ to $V'$ and obtain our $(17, (\varepsilon/100)^{3})$-\sprime $V$. 
By part 1 of \cref{lemma:ideal-cleanup}, we have that $\cG \subset \bC[V]$ and by part 2 we have $U + W \subset V_1$.
By \cref{lem:aux-controlFdeg} applied to $P_{1}, P_{2}$ and $P$ we deduce that either $P \in \bC\bs{V}$, or that $P$ is univariate over $V$.
In the latter case, the moreover part of \cref{lem:aux-controlFdeg} gives us the distinct linear forms $x_{1}, \dots, x_{t}$ and $a_{1}, \dots, a_{t}$.
Note that $t \geq \varepsilon^{3} m/40^3 - 4 (\varepsilon/100)^{3} m > (\varepsilon/100)^{3} m$.
As the choice of $P$ was arbitrary, the above holds for every $P \in \cR$, and therefore for every $P \in \cF_{\deg}$.
The last condition follows from the fact that \cref{lem:aux-controlFdeg} does not depend on how the clean vector space was constructed.

By part 3 of \cref{lemma:ideal-cleanup}, we have $\dim V = 128 \cdot (100/\varepsilon)^{3} + \dim(V_1') + 4 \cdot 17 \cdot \dim(V_2')$, which implies $\dim(V) = O(1/\varepsilon^3) = O(1/\delta^3)$, since $\dim(V_1') = \dim(U+W) = O(1/\delta)$ and $\dim(V_2') = |\cG| = O(1/\delta^2)$.
This concludes the proof.
\end{proof}

\subsection[Controlling F-span]{Controlling $\Fspan$}

We now construct a clean vector space $V$ where any form in $\Fspan$ is either in $\bC[V]$ or is univariate over $V$.
We shall use the previous vector spaces for $\Fsquare, \Flinear$ and $\Fdeg$ to help us in this case.

Before we prove the main lemma of this subsection, \cref{lem:controlFspan}, we need to understand the structure of the forms in $\Fspan$ relative to the current clean vector space that we have.
To that end, throughout this section, we will work with $W := W_1 + W_2$ being a $(17, (\varepsilon/100)^{3})$-clean vector space such that $\Fsquare \subset \bC[W]$, and the forms in $\Flinear \cup \Fdeg$ are either in the algebra $\bC[W]$ or are univariate over $W$.
We will also define $\cH$ as the set of forms in $\cF_{2}$ which are in $\bC[W]$ or univariate over $W$, and $\cG = \cF_{2} \setminus \cH$.
Note that by the properties of $W$ we have that $\cG \subseteq \Fspan$.

\begin{proposition}
\label{prop:fspan2far}
If $Q \in \Fspan$ is $2$-far from $W$, then one of the following holds: 
\begin{enumerate}
    \item there are at least $\varepsilon m / 3$ linear SG triples $(Q, F_i, G_i)$ where $F_i \in \cH$ and $G_i \in \cG$, where all forms $F_i, G_j$ are distinct 
    \item there are at least $\varepsilon m / 3$ distinct linear SG triples $(Q, F_i, G_i)$, where $F_i, G_j \in \cG$
\end{enumerate}
\end{proposition}

\begin{proof}
Since $Q$ is 2-far from $W$, $\Rank(Q - P) \geq 3$ for any $P \in \bC\bs{W}$.
Let $F \in \Fspan[Q]$ and $G \in  \Cspan{F, Q}\cap \cF_2$.
Both $F$ and $G$ cannot be in $\cH$: if that was the case, there would exist $P_F, P_G \in \bC[W]$ with $\Rank(F - P_{F}) \leq 1$ and $\Rank(G - P_{G}) \leq 1$, which implies $\Rank(Q - (P_{F} + P_{G})) \leq 2$, contradicting $Q$ being $2$-far from $W$.
Thus, either $F$ and $G$ are both in $\cG$, or one of them is in $\cH$ and the other is in $\cG$.

Suppose we have $F_1, F_2 \in \cH \cap \Fspan[Q]$ and
$G \in \Cspan{Q, F_1} \cap \Cspan{Q, F_2}\cap \cF_2$.
In this case, we would have $\Cspan{G, Q} = \Cspan{F_1, F_2}$, which implies $Q \in \Cspan{F_1, F_2}$, which contradicts $Q$ being $2$-far from $W$.
Thus, if there are at least $\varepsilon m / 3$ forms $F_1, \ldots, F_t$ in $\Fspan[Q] \cap \cH$, then there are at least $\varepsilon m / 3$ distinct forms $G_1, \ldots, G_t \in \cG$ such that $(Q, F_i, G_i)$ is a linear SG triple.

We are now left with the case where $|\Fspan[Q] \cap \cH| \leq \varepsilon m/3$.
Hence, $|\Fspan[Q] \cap \cG| \geq 2 \varepsilon m / 3$.
Let $P_1, \ldots, P_a$ be the forms in $\Fspan[Q] \cap \cG$ such that $\Cspan{Q, P_i}$ does not contain a third form in $\Fspan[Q] \cap \cG$.
If $R_{i}$ is the form spanned by $Q$ and $P_{i}$, we must have $R_{i} \in \cH$, and further it must be that $R_{i} \not\in (R_{j})$: if not then $P_{i}$ and $Q$ span $P_{j}$, contradicting our assumption on the $P_i$'s.

By the pigeon hole principle, and by the assumption that $\Fspan[Q] \cap \cH$ has less than $\varepsilon m / 3$ forms, it must be that $a \leq \varepsilon m /3$, and therefore that there are at least $\varepsilon m / 3$ forms in $\Fspan[Q] \cap \cG$ such that the form they span with $Q$ also lies in $\cG$.
\end{proof}

\begin{proposition}
\label{prop:fspan12close}
If $Q \in \Fspan$ is either $1$-close or $2$-close to $W$ and satisfies $\dim\br{\linvb{Q}{W}} \geq 2$ then there are at least $(\varepsilon - 2 (\varepsilon/100)^{3}) m$ distinct linear SG triples $(Q, F_i, G_i)$ with $F_i, G_i \in \cG$.
\end{proposition}

\begin{proof}

Let $Q = F_{Q} + E_{Q}$ with $F_{Q} \in \bC\bs{W}$ and $\Rank\br{E_{Q}} \leq 2$.
We have $\linvb{W}{Q} = \linvb{W}{E_{Q}}$.
Let $P \in \Fspan[Q]$ and let $R \in \cF_{2} \cap \Cspan{Q, P}$.
Let $Z := \lin{E_{Q}} + W$, so $Z_{1} = W_{1} + \lin{E_{Q}}$.

Assume both $P$ and $R$ are in $\cH$.
If either $P$ or $R$ is in $\bC\bs{W}$, then $Q$ is a univariate over $W$ which is a contradiction, so we have $\dim \br{\linvb{P}{W}} = \dim \br{\linvb{W}{R}} = 1$.
We can write $P = F_{P} + y_{P} z_{P}$ and $R = F_{R} + y_{R} z_{R}$ with $F_{P}, F_{R} \in \bC\bs{W}$.
Let $R = \alpha P + \beta Q$ for some $\alpha,\beta \in \bC^*$. We have $F_{Q} +\alpha F_{P} -\beta F_{R} = \beta y_{R} z_{R} - \alpha y_{P} z_{P} - E_{Q}$, and this form is in $\bC\bs{W}$.
Since the rank of $\beta y_{R}z_{R} -\alpha y_{P} z_{P} - E_{Q}$ is bounded by $4$, this form is in $\bC\bs{W_{1}}$.
Combining this with the fact that $y_{P}, y_{R}$ are spanned by $z_{P}, z_{R}$ in $S_{1} / W_{1}$ we get $E_{Q} \in \bC\bs{W_{1}, z_{P}, z_{R}}$ and $\lin{E_{Q}} \subseteq W_{1} + \langle z_{P}, z_{R} \rangle$.
Since $\dim\br{\linvb{W}{Q}} \geq 2$, the image of $\lin{E_{Q}}$ in $S_{1} / W_{1}$ has dimension at least $2$, and therefore is spanned by $z_{P}, z_{R}$.
We then get $z_{P}, z_{R} \in \linvb{W}{Q}$, and that $\dim \br{\linvb{Z}{P} } = \dim \br{\linvb{Z}{R}} = 0$.

Suppose now that $P \in \cH$ and $R \in \cG$.
We again write $P = F_{P} + y_{P} z_{P}$, where $y_{P} = z_{P} = 0$ if $P \in \bC\bs{W}$.
Since $Q$ is at most $2$-close and since $P$ is at most $1$-close, the form $R$ is at most $3$-close from $W$.
We can therefore write $R = F_{R} + E_{R}$, with $F_{R} \in \bC\bs{W}$ and $\Rank \br{E_{R}} \leq 3$.
Since $R \not \in \cH$ we have $\dim\br{\linvb{W}{R}} \geq 2$.
Combining the equations again we get $E_{R} - E_{Q} - y_{P} z_{P} \in \bC\bs{W_{1}}$, and therefore $E_{Q} \in \bC\bs{W_{1}, z_{P}, \lin{E_{R}}}$.
This implies that $\linvb{W}{Q} \subseteq \langle z_{P} \rangle + \linvb{W}{R}$.
Combining this with the facts that $\dim\br{\linvb{W}{R}}, \dim\br{\linvb{W}{Q}} \geq 2$ we get that $\dim\left( \br{\linvb{W}{Q}} \cap \br{\linvb{W}{R}} \right) \geq 1$.
This in turn implies that $\dim\br{\linvb{W}{R}} > \dim\br{\linvb{Z}{R}}$.
In the case when $P \in \cG$ and $R \in \cH$, we can repeat the above analysis to deduce $\dim\br{\linvb{W}{P}} > \dim\br{\linvb{Z}{P}}$.

We now use the clean condition to bound how many times the above two cases can occur.
Let $P_{1}, \dots, P_{a}$ be the forms in $\Fspan[Q]$, and let $R_{i} \in \cF_{2} \cap \Cspan{Q, P_{i}}$.
Let $P_{1}, \dots, P_{b}$ be the forms such that $P_{i} \in \cH$ and $R_{i} \in \cG$.
For every $i \leq b$ we have $\dim\br{\linvb{W}{R_{i}}} > \dim\br{\linvb{Z}{R_{i}}}$ by the above.
If $R_{i} = R_{j}$ then $Q \in \Cspan{P_{i}, P_{j}}$.
From the above analysis, this implies $\dim \br{\linvb{W}{P_{i}}} = \dim \br{\linvb{W}{P_{j}}} = 1$, while $\dim \br{\linvb{Z}{P_{i}}} = \dim \br{\linvb{Z}{P_{j}}} = 0$.
Therefore, among the $2b$ forms $P_{1}, \dots, P_{b}, \allowbreak R_{1}, \dots, R_{b}$, there are at least $b$ distinct forms whose relative space of linear forms with respect to $Z$ has lower dimension than that with respect to $W$, and by the clean condition $b \leq (\varepsilon/100)^{3} m$.
Let $P_{b+1}, \dots, P_{b'}$ be the forms with $R_{i} \in \cH$.
By the above analysis, for every such $i$ we have $\dim\br{\linvb{W}{P}} > \dim\br{\linvb{Z}{P}}$, whence $b' - b \leq (\varepsilon/100)^{3} m$.
For every $P_{i}$ with $i > b'$ we must have $P_{i}, R_{i} \in \cG$, completing the proof.
\end{proof}

\begin{lemma}\label{lem:controlFspan}
Let $W := W_1 + W_2$ be a $(17, (\varepsilon/100)^{3})$-clean vector space such that $\Fsquare \subset \bC[W]$, and the forms in $\Fdeg \cup \Flinear$ are either in $\bC[W]$ or are univariate over $W$.
Then there exists a $(17, (\varepsilon/100)^{3})$-clean vector space $U := U_1 + U_2$ with $\bC\bs{W} \subset \bC\bs{U}$ such that every form in $\Fspan$ is either in $\bC\bs{U}$ or is univariate over $U$.
Moreover, $\dim U_{2} = \dim W_{2} + \bigO{1 / \varepsilon}$ and $\dim U_{1} = \dim W_{1} +  \bigO{\dim W_{2} / \varepsilon + 1 / \varepsilon^{4}}$.
\end{lemma}

\begin{proof}
    We use an iterative process to construct $U$.
    Start by setting $U^{(0)} := W$.
    Let $\cH^{(0)}$ consist of all forms univariate over $U^{(0)}$ and $\cG^{(0)} := \cF_{2} \setminus \cH^{(0)}$.
    Suppose $Q \in \cG^{(0)}$ satisfies case 1 of \cref{prop:fspan2far}, that is, there are at least $\varepsilon m / 3$ linear SG triples $\br{Q, F_{i}, G_{i}}$ with $F_{i} \in \cH^{(0)}$ and $G_{i} \in \cG^{(0)}$.
    We then set $U'^{(1)} := U^{(0)} + \Cspan{Q}$, and set $U^{(1)}$ to be the result of applying the clean up procedure (\cref{lemma:ideal-cleanup}) to $U'^{(1)}$ with parameters $r, (\varepsilon/100)^{3}$.
    We also set $\cH^{(1)}$ to be the subset of $\cF_{2}$ that is univariate over $U^{(1)}$, and set $\cG^{(1)} := \cF_{2} \setminus \cH^{(1)}$.
    The form $Q$ is in $\bC\bs{U^{(1)}}$, and the forms $G_{i}$ are univariate over $U^{(1)}$, and therefore $\abs{\cH^{(1)}} - \abs{\cH^{(0)}} \geq \varepsilon m / 3$.
    By \cref{lemma:ideal-cleanup}, we have $\dim U^{(1)}_{2} \leq \dim U^{(0)}_{2} + 1$ and 
    $$\dim U^{(1)}_{1} \leq \dim U^{(0)}_{1} + 128 \cdot (100/\varepsilon)^{3} + 68 \cdot \br{\dim U^{(0)}_{2} + 1}.$$
    We repeat the above process as long as we can find such a form $Q$.
    After $t$ steps, we have vector space $U^{(t)}$ and subset $\cH^{(t)}$ such that $\abs{\cH^{(t)}} - \abs{\cH^{(0)}} \geq \varepsilon t m / 3$, and $\dim U^{(t)}_{2} \leq \dim U^{(0)}_{2} + t$, and
    $$\dim U^{(t)}_{1} \leq \dim U^{(0)}_{1} + 128 \cdot (100/\varepsilon)^{3} \cdot t + 68 \cdot t \cdot \dim U^{(0)}_{2} + 68 \cdot \frac{t(t+1)}{2}.$$
    Since $\abs{\cH^{(t)}} \leq m$, the process must stop after $t = 3 / \varepsilon$ steps.
    
    Once we can no longer find such a form $Q$, every form in $\cG^{(t)}$ satisfies either case $2$ of \cref{prop:fspan2far} or \cref{prop:fspan12close}, and therefore $\cG^{(t)}$ is a robust linear SG configuration with parameter $\varepsilon / 6$.
    By \cref{thm:DSW14} there are forms $P_{1}, \dots, P_{b}$ with $b \leq 78 / \varepsilon$ such that $\cG^{(t)} \subseteq \Cspan{P_{1}, \dots, P_{b}}$.
    We set $U' := U^{(t)} + \Cspan{P_{1}, \dots, P_{b}}$, and set $U$ to be the result of the clean up procedure (\cref{lemma:ideal-cleanup}) to $U'$ with parameters $r, \varepsilon^{3}/4^{8}$.
    By construction, every form in $\Fspan$ is either in $\bC\bs{U}$ or univariate over $U$.
    Since $\bC\bs{W} \subseteq \bC\bs{U}$, every form in $\cF$ is also either in $\bC\bs{U}$ or univariate over $U$.
    Finally, using the bound on $\dim U^{(t)}$ and $b$ we get $\dim U_{2} \leq \dim W_{2} + C / \varepsilon$ and
    $$\dim U_{1} \leq \dim W_{1} + C \cdot \frac{\dim W_{2}}{\varepsilon} + C \cdot \frac{1}{\varepsilon^{4}}$$
    for some universal constant $C$.
\end{proof}

\subsection{Proof of Theorem \ref{thm:deltasg2}}

We now prove our main theorem: $\drsg{\delta}{2}$ configurations must lie in a constant dimensional vector space.
We prove our result by first constructing an auxiliary algebra where the forms of our configuration $\cF$ are either in the algebra, or are univariate over this algebra, and then prove that we can augment this algebra slightly to contain all forms from $\cF$.
The first step follows from the lemmas we have proved in the last four subsections.
The second step is to show that extra variables corresponding to the forms form a LCC configuration, allowing us to bound their rank.

\begin{lemma}[Reduction to Base Configuration]\label{lemma:reduction-to-basic}
Let $0 < \delta \leq 1$ be a constant, and let $\varepsilon := \delta / 10$.
Let $\cF$ be a $\drsg{\delta}{2}$ configuration. 
There exists a $(17, (\varepsilon/100)^{3})$-\sprime with respect to $\cF$, denoted by $V$, such that: 
\begin{enumerate}
    \item every form in $\cF$ is either in $\bC\bs{V}$ or univariate over $V$, and $\dim(V) = O(1/\varepsilon^{4})$.
    \item $\Fsquare \cup \Flinear \subseteq \bC\bs{V}$. 
    \item For every form $P \in \Fdeg \setminus \bC\bs{V}$, if $\linvb{V}{P} = \Cspan{z_P}$ then there are at least $(\varepsilon/100)^{3}$ distinct linear forms $x_{1}, \dots, x_{t}\in \cF_1$ and distinct linear forms $a_{1}, \dots, a_{t}\in \cF_1$ such that for every $i$, the linear forms $z_{P}, \overline{x_{i}}, \overline{a_{i}}$ are pairwise linearly independent in $S_{1} / V_{1}$, and $z_{P} \in \Cspan{\overline{x_{i}}, \overline{a_i}}$. 
\end{enumerate}
\end{lemma}

\begin{proof}

Let $V^{(1)}$ be the vector space obtained by applying \cref{lem:controlFLsquare} to $\cF$.
Let $V^{(2)}$ be the vector space obtained by applying \cref{proposition:flinear-ps20} to $\cF$.
Let $V^{(3)}$ be the vector space obtained by applying the clean up procedure (\cref{lemma:ideal-cleanup}) to $V^{(2)}$, with parameters $17, (\varepsilon/100)^{3}$.
We have $\dim V^{(3)}_{1} = \bigO{1 / \varepsilon^{3}}$ and $\dim V^{(3)}_{2} = \bigO{1 / \varepsilon^{2}}$.
By \cref{lem:controlFLsquare} we have $\Fsquare \subseteq \bC\bs{V^{(3)}}$, and by \cref{lem:linstructure} every form in $\Flinear$ is either in $\bC\bs{V^{(3)}}$ univariate over $V^{(3)}$.

Now let $V^{(4)}$ be the $(17, (\varepsilon/100)^{3})$-clean vector space obtained by applying \cref{lem:controlFdeg} to $\Fdeg$.
Let $V^{(5)}$ be the result of applying the clean up procedure (\cref{lemma:ideal-cleanup}) to $V^{(3)} + V^{(4)}$ with parameters $17, (\varepsilon/100)^{3}$.
We have $\dim V^{(5)}_{1} = \bigO{1 / \varepsilon^{3}}$ and $\dim V^{(5)}_{2} = \bigO{1 / \varepsilon^{3}}$. 
Every form in $\Fdeg$ is now either in $\bC\bs{V^{(5)}}$ or univariate over $V^{(5)}$.
The conditions on $\Fsquare, \Flinear$ still hold with $V^{(5)}$ in place of $V^{(3)}$.

Finally, let $V$ be the result of applying \cref{lem:controlFspan} to $V^{(5)}$.
Every form in $\Fspan$ is either in $\bC\bs{V}$ or univariate over $V$.
The vector space $V$ is $(17, (\varepsilon/100)^{3})$-clean, and we have $\dim V_{1} = \bigO{1 / \varepsilon^{4}}$ and $\dim V_{2} = \bigO{1 / \varepsilon^{3}}$.

By construction, we have $\Fsquare \subseteq \bC\bs{V}$ and that every form in $\cF$ is either in $\bC\bs{V}$ or univariate over $V$.
Note that this condition trivially holds for the forms in $\cF_{1}$.
We can now apply \cref{lem:controlFLinear} to deduce that $\Flinear \subseteq \bC\bs{V}$.
Since $\bC\bs{V^{(5)}} \subseteq \bC\bs{V}$, the condition for forms in $\Fdeg \setminus \bC\bs{V}$ holds by the last statement of \cref{lem:controlFdeg}.
\end{proof}

\begin{lemma}[Base Configuration]\label{lemma:basic-sg-configuration}
Let $0 < \delta \leq 1$, $\varepsilon := \delta / 10$ be constants.
If $\cF$ is a $\drsg{\delta}{2}$ configuration, and $V := V_1 + V_2$ is  $(17, (\varepsilon/100)^{3})$-\sprime with $\cF$ that satisfies the conditions of \cref{lemma:reduction-to-basic}, then there exists $U \subset S_1$ with $\dim(U) = O(1/\delta^{27})$ such that $\cF \subset \bC[V, U]$.  
\end{lemma}

\begin{proof}

For every $Q \in \cF_{2}$, we use $z_{Q}$ to denote the linear form in $S_{1} / V_{1}$ that spans $\linvb{V}{Q}$.
For linear forms $\ell \in \cF_{1}$, we use $z_{\ell}$ to denote the image of $\ell$ in $S_{1} / V_{1}$.
If a form is in $\bC\bs{V}$ then $z_{Q} = 0$, and if a linear form is in $V_{1}$ then $z_{\ell} = 0$.
Let $\lambda := \frac{1}{4} \cdot (\varepsilon/100)^{3}$.
We now show that the set of nonzero linear forms $z_{Q}, z_{\ell}$ form a $\lambda$-LCC configuration, which allows us to bound their dimension.

Let $\cF'$ be the subset of $\cF$ that consists of forms not in $\bC\bs{V}$.
Define $\cF'_{1} := \cF' \cap S_{1}$ and $\cF'_{2} = \cF' \cap S_{2}$.
Let $P_{1}, \dots, P_{b}$ be the forms in $\cF'_{2}$ and let $\ell_{b+1}, \dots, \ell_{c}$ be the forms in $\cF'_{1}$.
By definition, the linear forms $z_{P_{i}}$ and $z_{\ell_{j}}$ are all nonzero, although they might not all be distinct.
In order to show that these linear forms form a $\lambda$-LCC configuration, we have to show the following: given a form $P \in \cF'$ and a subset $\Gamma \subset \cF'$ of size at most $\lambda m$, there are forms $F, G \in \cF' \setminus \Gamma$ such that $z_{P} \in \Cspan{z_{F}, z_{G}}$.

Suppose $P \in \cF'_{2}$.
Since $\Flinear, \Fsquare \subseteq \bC\bs{V}$ we have $P \in \Fspan \cup \Fdeg$.
We handle each of these two cases separately.
Let $Z := V + z_{P}$.
\begin{itemize}
    \item Suppose $P \in \Fspan$, and suppose $\br{P, F_{i}, G_{i}}$ are the SG triples corresponding to $P$.
    If $F_{i} \in \bC\bs{V}$, then $z_{P} \in \ideal{z_{G_{i}}}$.
    If $G_{i} = G_{j}$ then $P \in \Cspan{F_{i}, F_{j}}$, contradicting $P \not \in \bC\bs{V}$.
    This implies that $G_{i}$ and $G_{j}$ are distinct if $F_{i}, F_{j} \in \bC\bs{V}$.
    Therefore, by the clean condition there are at most $(\varepsilon/100)^{3} m$ SG triples with $F_{i} \in \bC\bs{V}$ since for each of these triples we have $\dim \linvb{V}{G_{i}} = 1$ and $\dim \linvb{Z}{G_{i}} = 0$.
    The clean condition also implies that there are at most $(\varepsilon/100)^{3} m$ SG triples with $z_{F_{i}} \in \ideal{z_{P}}$, since for each of these we have $\dim \linvb{V}{F_{i}} = 1$ and $\dim \linvb{Z}{F_{i}} = 0$.
    
    We only consider the remaining $(\varepsilon - 2(\varepsilon/100)^{3}) m$ triples, where we have $z_{F_{i}} \neq 0$ and $z_{F_{i}} \not \in  \ideal{z_{P}}$.
    The latter condition also implies $G_{i} \not \in \bC\bs{V}$.
    Therefore, every triple we consider satisfies $F_{i}, G_{i} \in \cF'$.
    Let $\Gamma$ be any subset of $\cF'$ of size at most $\lambda m$.
    After removing at most $\lambda m$ SG triples, since $\varepsilon - 2 (\varepsilon/100)^{3} -\lambda > 2\lambda$ there are still $2\lambda m$ triples with $F_{i} \not \in \Gamma$.
    If $G_{i} = G_{j}$, for two of such triples, then $P \in \Cspan{F_{i}, F_{j}}$ and we are done.
    If the forms $G_{i}$ are all distinct for these $2 \lambda m$ triples, then after removing at most $\lambda m$ more triples we can also assume that $G_{i} \not \in \Gamma$ for any triple.
    In either case, $P$ is spanned by two forms $Q_{1}, Q_{2} \in \cF' \setminus \Gamma$, and $z_{P}$ is spanned by $z_{Q_{1}}, z_{Q_{2}}$.
    Hence $z_P$ satisfies the condition of the $\lambda$-LCC configuration.
    
    \item Suppose now that $P \in \Fdeg$.
    By assumption, there are at least $(\varepsilon/100)^{3} m$ linear forms $x_{1}, \dots, x_{t} \in \cF'_{1}$, and linear forms $a_{1}, \dots, a_{t} \in \cF'_{1}$ such that $z_{P} \in \Cspan{z_{x_{i}}, z_{a_{i}}}$.
    \cref{lem:aux-controlFdeg} also guarantees that each of the $x_{i}$ are distinct, and that each of the $a_{i}$ are distinct.
    Let $\Gamma \subset \cF'$ of size at most $\lambda m$.
    Since the $x_{i}$ are distinct, and since the $a_{i}$ are distinct, there are at most $2 \lambda m$ indices with either $x_{i} \in \Gamma$ or $a_{i} \in \Gamma$.
    Since $(\varepsilon/100)^{3} m - 2 \lambda m > 0$, we can find an index $i$ such that $z_{P} \in \Cspan{z_{x_{i}}, z_{a_{i}}}$ with $x_{i}, a_{i} \in \cF'_{1} \setminus \Gamma$, proving that $z_P$ satisfies the $\lambda$-LCC criterion.
\end{itemize}

Now let $\ell \in \cF'_{1}$ and $\Gamma \subset \cF'$ such that $|\Gamma| \leq \lambda m$.

\begin{itemize}
    \item Suppose $\ell$ satisfies the SG condition with at least $\varepsilon m / 2$ linear forms, that is, there are at least $\varepsilon m / 2$ SG triples $(\ell, g_{i}, H_{i})$ with $g_{i} \in \cF_{1}$.
    If $H_{i} \in \cF_{2}$, then $\ell \in \lin{H_{i}}$, and since $H_{i}$ is univariate over $V$ we get $z_{H_{i}} \in \ideal{z_{\ell}}$.
    By the clean condition, there can be at most $(\varepsilon/100)^{3} m$ triples with $H_{i} \in \cF_{2}$, since $\dim \linvb{V}{H_{i}} > \dim \linvb{V + \bc{z_{\ell}}}{H_{i}}$.
    We therefore restrict our attention to the $\varepsilon m / 2 - (\varepsilon/100)^{3} m$ triples $\br{\ell, g_{i}, H_{i}}$ where $H_{i}$ is a linear form.
    By the same argument as the case of $P \in \Fspan$ above we get either $\ell \in \Cspan{g_{i}, g_{j}}$ with $g_{i}, g_{j} \in \cF'_{1} \setminus \Gamma$ or $\ell \in \Cspan{g_{i}, H_{i}}$ with $g_{i}, H_{i} \in \cF'_{1} \setminus \Gamma$.
    This step requires $(\varepsilon  / 2 - 3 (\varepsilon/100)^{3}  - 2 \lambda) m > 0$ which holds by the choice of $\lambda$.
    \item Finally suppose $\ell$ satisfies the SG condition with at least $\varepsilon m / 2$ quadratics, that is, there are at least $\varepsilon m / 2$ triples $(\ell, F_{i}, H_{i})$ with $F_{i} \in \cF_{2}$.
    By \cref{proposition:irreducible-radical-lin}, the ideal $\ideal{\ell, F_{i}}$ is not radical only if $\ell \in \lin{F_{i}}$.
    Since $F_{i}$ is univariate over $V$, and since $\ell \not \in V_{1}$, this condition implies $z_{F_{i}} \in \ideal{z_{\ell}}$.
    By the clean condition, this can happen for at most $(\varepsilon/100)^{3} m$ triples, since $\dim \linvb{V}{F_{i}} > \dim \linvb{V + \bc{z_{\ell}}}{F_{i}}$.
    For the remaining triples, we have $H_{i} = F_{i} + \ell h_{i}$ for some linear form $h_{i}$.
    If $F_{i} \in \bC\bs{V}$, then $\linvb{V}{H_{i}}$ is spanned by $\ell, h_{i}$, and since $H_{i}$ is univariate over $V$, we get $z_{H_{i}} \in \ideal{z_{\ell}}$.
    If $F_{i}, F_{j} \in \bC\bs{V}$, then $H_{i} \neq H_{j}$, otherwise $0 \neq \ell(h_{i} - \mu h_{j}) \in \bC\bs{V_{1}}$ for some constant $\mu \in \bC$, contradicting $\ell \not \in V_{1}$.
    Therefore, by the clean condition, there are at most $(\varepsilon/100)^{3} m$ triples with $F_{i} \in \bC\bs{V}$, since for each such index we have $\dim \linvb{V}{H_{i}} > \dim \linvb{V + \bc{z_{\ell}}}{H_{i}}$.
    For every other triple we have $\ell h_{i} + z_{F_{i}} y_{F_{i}} - z_{H_{i}} y_{H_{i}} \in \bC\bs{V_{1}}$, whence $z_{\ell} \in \Cspan{z_{F_{i}}, z_{H_{i}}}$.
    By the same argument as the case of $P \in \Fspan$ above we get either $z_{\ell} \in \Cspan{z_{F_{i}}, z_{F_{j}}}$ with $F_{i}, F_{j} \in \cF'_{2} \setminus \Gamma$ or $z_{\ell} \in \Cspan{z_{F_{i}}, z_{H_{i}}}$ with $F_{i}, H_{i} \in \cF'_{2} \setminus \Gamma$.
    This step also requires $(\varepsilon  / 2 - 3 (\varepsilon/100)^{3}  - 2 \lambda) m > 0$ which holds by the choice of $\lambda$.
\end{itemize}

Therefore, the set $\{ z_{P_{i}}, z_{\ell_{j}} \mid P_{i}, \ell_{j} \in \cF'\}$ forms a $\lambda$-LCC configuration.
By \cref{thm:BDWY11}, there is a vector space $U$ of $S_{1} / V_{1}$ of dimension $\bigO{\lambda^{9}}$ such $z_{P_{i}}, z_{\ell_{j}} \in U$.
Since $z_{P} = 0$ for every $P \in \cF \setminus \cF'$, we can say $z_{P} \in U$ for every $P \in \cF$.
Finally, since every form $P \in \cF$ is univariate over $V$, we have $P \in \bC\bs{V, z_{P}}$, whence $P \in \bC\bs{V, U}$ as required.
\end{proof}

\deltasgstatement*

\begin{proof}
    Let $\varepsilon := \delta / 10$.
    Given a $\drsg{\delta}{2}$, apply \cref{lemma:reduction-to-basic} to obtain $V$, a $(17, (\varepsilon/100)^{3})$-clean vector space with respect to $\cF$ such that $\dim V = \bigO{1 / \varepsilon^{4}}$, and every form in $\cF$ is either in $\bC\bs{V}$, or univariate over $V$.
    We now apply \cref{lemma:basic-sg-configuration} with vector space $V$, to obtain a vector space $U \subseteq S_{1}$ such that $\dim(U) = \bigO{1 / \delta^{27}}$, and $\cF \subseteq \bC\bs{V, U}$.
    
    Since the generators of $\bC\bs{V, U}$ are homogeneous, the set of linear forms $\bC\bs{V, U}_{1}$ in the vector space $U + V_{1}$.
    Further, every quadratic in this algebra is a linear combination of elements of $V_{2}$ and products of the form $\ell_{1} \ell_{2}$, where $\ell_{i} \in U + V_{1}$.
    Therefore, $\dim \bC\bs{V, U}_{2} = \bigO{1 / \delta^{54}}$.
    The vector space $\bC\bs{V, U}_{1} + \bC\bs{V, U}_{2}$ contains $\cF$ and has dimension $\bigO{1 / \delta^{54}}$.
\end{proof}

\section{Conclusion and Open Problems}\label{sec:conclusion}

In this paper, we prove a robust version of the radical Sylvester-Gallai theorem for quadratics, generalizing \cite{S19}.
Just as in the linear case of the Sylvester-Gallai problem, robustness plays an important role in generalizing Sylvester-Gallai results to higher dimensional variants, such as the flats version in \cite{BDWY11}. 
We similary expect our robust variant to allow us to generalize the Sylvester-Gallai problem to ``higher codimension'' tuples of quadratic polynomials. 
For instance, instead of requiring $\radideal{F_i, F_j}$ to intersect $\cF$ non-trivially, one would only require that for many triples $(i,j,k)$, we would require $\radideal{F_i, F_j, F_k}$ to intersect $\cF$ non-trivially.
Just as in the linear case, properly defining such higher codimension variants requires some careful thought, especially since the non-linear aspect will introduce more subtlety than the linear case.
These higher dimensional variants have applications in algebraic complexity, as they can be instrumental in proving the main conjectures posed in \cite{gupta2014algebraic} about such SG configurations.

Another important open problem is to generalize the above result to prove a robust version of the ``product version'' of the Sylvester-Gallai problem - a robust version of \cite[Conjecture 1]{gupta2014algebraic} with $k=3$ and $r=2$.
In this work, we made a somewhat strong use of the fact that we have an extra polynomial in the radical ideal, and having a product of polynomials in the ideal instead seems to require a strengthening of several arguments in this paper to address it.
Just as in \cite{PS20a}, we believe that our general structure theorem, which gives us a deeper look in the minimal primes, could shed some light into a different way to construct robust algebras.

It is important to remark that higher codimension variants of the Sylvester-Gallai problem, even for quadratics, involves the study of schemes which are not equidimensional, which may require stronger structural results on the structure of such ideals.
However, one could hope that our structure theorems might suffice, just as in \cite{BDWY11} the robust linear Sylvester-Gallai theorem was sufficient to induct on the higher-dimensional analogs.

Another interesting direction and potential application of robust SG configurations is in the study of non-linear locally correctable codes (LCCs) over fields of characteristic zero. 
While lower bounds for linear LCCs have been out of reach for current techniques even over characteristic zero,\footnote{Aside from 2-query LCCs where optimal lower bounds are known for both linear and non-linear codes} 
it would be interesting to know if robust non-linear SG configurations have bounded transcendence degree.
If a robust form of Gupta's general conjecture is false, it could yield the first constructions of non-linear LCCs with non-constant dimension over characteristic zero, which are not known to exist. 

\bibliographystyle{alpha}
\bibliography{main}

\newcommand{\etalchar}[1]{$^{#1}$}
\begin{thebibliography}{CTSSD87}

\bibitem[AM69]{AM69}
M.~F. Atiyah and I.~G. MacDonald.
\newblock {\em Introduction to Commutative Algebra}.
\newblock Addison Wesley Publishing Company, 1969.

\bibitem[BDYW11]{BDWY11}
Boaz Barak, Zeev Dvir, Amir Yehudayoff, and Avi Wigderson.
\newblock Rank bounds for design matrices with applications to combinatorial
  geometry and locally correctable codes.
\newblock In {\em Proceedings of the Forty-Third Annual ACM Symposium on Theory
  of Computing}, STOC '11, page 519–528, New York, NY, USA, 2011. Association
  for Computing Machinery.

\bibitem[BM90]{borwein1990survey}
Peter Borwein and William~OJ Moser.
\newblock A survey of sylvester's problem and its generalizations.
\newblock {\em Aequationes Mathematicae}, 40(1):111--135, 1990.

\bibitem[CKS19]{CKS19}
Chi-Ning Chou, Mrinal Kumar, and Noam Solomon.
\newblock Closure results for polynomial factorization.
\newblock {\em Theory of Computing}, 15(1):1--34, 2019.

\bibitem[CTSSD87]{CTSSD87}
J-L Colliot-Th{\'e}lene, J-J Sansuc, and P~Swinnerton-Dyer.
\newblock Intersections of two quadrics and ch{\^a}telet surfaces. i.
\newblock {\em Journal f{\"u}r die reine und angewandte Mathematik},
  373:37--107, 1987.

\bibitem[DDS21]{DDS21}
Pranjal Dutta, Prateek Dwivedi, and Nitin Saxena.
\newblock Deterministic identity testing paradigms for bounded top-fanin
  depth-4 circuits.
\newblock In {\em Proceedings of the 36th Computational Complexity Conference},
  CCC '21, Dagstuhl, DEU, 2021. Schloss Dagstuhl--Leibniz-Zentrum fuer
  Informatik.

\bibitem[Dic14]{dickson1914points}
Leonard~Eugene Dickson.
\newblock The points of inflexion of a plane cubic curve.
\newblock {\em The Annals of Mathematics}, 16(1/4):50--66, 1914.

\bibitem[DSW14]{DSW14}
Zeev Dvir, Shubhangi Saraf, and Avi Wigderson.
\newblock Improved rank bounds for design matrices and a new proof of kelly’s
  theorem.
\newblock In {\em Forum of Mathematics, Sigma}, volume~2. Cambridge University
  Press, 2014.

\bibitem[Dvi12]{dvir2012incidence}
Zeev Dvir.
\newblock Incidence theorems and their applications.
\newblock {\em arXiv preprint arXiv:1208.5073}, 2012.

\bibitem[EBW{\etalchar{+}}43]{erdos1943problems}
Paul Erdos, Richard Bellman, Hubert~S Wall, James Singer, and Victor
  Th{\'e}bault.
\newblock Problems for solution: 4065-4069.
\newblock {\em The American Mathematical Monthly}, 50(1):65--66, 1943.

\bibitem[EG84]{EG84}
D.~Eisenbud and S.~Goto.
\newblock Linear free resolutions and minimal multiplicity.
\newblock {\em J.Algebra}, 88:89–133, 1984.

\bibitem[Eis95]{Eis95}
David Eisenbud.
\newblock {\em Commutative Algebra with a View Toward Algebraic Theory.}
\newblock Springer-Verlag, New York, 1995.

\bibitem[Eng07]{Eng07}
B.~Engheta.
\newblock On the projective dimension and the unmixed part of three cubics.
\newblock {\em J.Algebra}, 316:715–734, 2007.

\bibitem[Gal44]{gallai1944solution}
Tibor Gallai.
\newblock Solution of problem 4065.
\newblock {\em American Mathematical Monthly}, 51:169--171, 1944.

\bibitem[Gup14]{gupta2014algebraic}
Ankit Gupta.
\newblock Algebraic geometric techniques for depth-4 pit \& sylvester-gallai
  conjectures for varieties.
\newblock In {\em Electron. Colloquium Comput. Complex.}, volume~21, page 130,
  2014.

\bibitem[Han65]{hansen1965generalization}
Sten Hansen.
\newblock A generalization of a theorem of sylvester on the lines determined by
  a finite point set.
\newblock {\em Mathematica Scandinavica}, 16(2):175--180, 1965.

\bibitem[Hir83]{H83}
Friedrich Hirzebruch.
\newblock Arrangements of lines and algebraic surfaces.
\newblock In {\em Arithmetic and geometry}, pages 113--140. Springer, 1983.

\bibitem[HP94]{HP94}
William Vallance~Douglas Hodge and Daniel Pedoe.
\newblock {\em Methods of Algebraic Geometry: Volume 2}.
\newblock Cambridge University Press, 1994.

\bibitem[Kel86]{kelly1986resolution}
Leroy~Milton Kelly.
\newblock A resolution of the sylvester-gallai problem of j.-p. serre.
\newblock {\em Discrete \& Computational Geometry}, 1(2):101--104, 1986.

\bibitem[KS09]{kayal2009blackbox}
Neeraj Kayal and Shubhangi Saraf.
\newblock Blackbox polynomial identity testing for depth 3 circuits.
\newblock In {\em 2009 50th Annual IEEE Symposium on Foundations of Computer
  Science}, pages 198--207. IEEE, 2009.

\bibitem[LST21]{LimayeST21}
Nutan Limaye, Srikanth Srinivasan, and S{\'{e}}bastien Tavenas.
\newblock Superpolynomial lower bounds against low-depth algebraic circuits.
\newblock {\em Electron. Colloquium Comput. Complex.}, page~81, 2021.

\bibitem[Mel40]{melchior1940uber}
Eberhard Melchior.
\newblock Uber vielseite der projektiven ebene.
\newblock {\em Deutsche Math}, 5:461--475, 1940.

\bibitem[MM18]{MM18}
Paolo Mantero and Jason McCullough.
\newblock A finite classification of (x, y)-primary ideals of low multiplicity.
\newblock {\em Collectanea Mathematica}, 69(1):107--130, 2018.

\bibitem[OS22]{OS21}
Rafael Oliveira and Akash Sengupta.
\newblock Radical sylvester-gallai theorem for cubics.
\newblock {\em Manuscript}, 2022.

\bibitem[PS20a]{PS20a}
Shir Peleg and Amir Shpilka.
\newblock A generalized sylvester-gallai type theorem for quadratic
  polynomials.
\newblock {\em CoRR}, abs/2003.05152, 2020.

\bibitem[PS20b]{PS20b}
Shir Peleg and Amir Shpilka.
\newblock Polynomial time deterministic identity testing algorithm for
  {\(\Sigma\)}\({}^{\mbox{[3]}}\){\(\Pi\)}{\(\Sigma\)}{\(\Pi\)}\({}^{\mbox{[2]}}\)
  circuits via edelstein-kelly type theorem for quadratic polynomials.
\newblock {\em CoRR}, abs/2006.08263, 2020.

\bibitem[PS22]{PS22}
Shir Peleg and Amir Shpilka.
\newblock Robust sylvester-gallai type theorem for quadratic polynomials.
\newblock {\em CoRR}, abs/2202.04932, 2022.

\bibitem[Ser66]{serre1966advanced}
Jean-Pierre Serre.
\newblock Advanced problem 5359.
\newblock {\em Amer. Math. Monthly}, 73(1):89, 1966.

\bibitem[Shp20]{S19}
Amir Shpilka.
\newblock Sylvester-gallai type theorems for quadratic polynomials.
\newblock {\em Discrete Analysis}, page 14492, 2020.

\bibitem[Sin16]{sinha2016reconstruction}
Gaurav Sinha.
\newblock Reconstruction of real depth-3 circuits with top fan-in 2.
\newblock In {\em 31st Conference on Computational Complexity (CCC 2016)}.
  Schloss Dagstuhl-Leibniz-Zentrum fuer Informatik, 2016.

\bibitem[SS13]{saxena2013sylvester}
Nitin Saxena and Comandur Seshadhri.
\newblock From sylvester-gallai configurations to rank bounds: Improved
  blackbox identity test for depth-3 circuits.
\newblock {\em Journal of the ACM (JACM)}, 60(5):1--33, 2013.

\bibitem[ST83]{szemeredi1983extremal}
Endre Szemer{\'e}di and William~T. Trotter.
\newblock Extremal problems in discrete geometry.
\newblock {\em Combinatorica}, 3(3):381--392, 1983.

\bibitem[Syl93]{sylvester1893mathematical}
James~Joseph Sylvester.
\newblock Mathematical question 11851.
\newblock {\em Educational Times}, 59(98):256, 1893.

\end{thebibliography}

\appendix

\section{Elaborating on the proof of Proposition~\ref{proposition:xy-primary} }

The results of \cite{MM18} classify the low degree parts of $\ideal{x, y}$ primary ideals of degrees at most $4$.
We are interested in the ideals primary to $\ideal{x, y}$ that can contain two irreducible quadratics $Q_{1}, Q_{2}$ such that $\lin{Q_{1}} \neq \lin{Q_{2}}$.
Suppose $J$ is an ideal primary to $\ideal{x, y}$ that contains $Q_{1}$ and $Q_{2}$.
The degree of $\ideal{Q_{1}, Q_{2}}$ is at most $4$, and therefore the same is true of $J$.
Since $Q_{1}, Q_{2}$ have degree two, and since $J$ is homogeneous, we are only interested in the generators of $J$ of degree at most $2$.
Note that every irreducible homogeneous forms in two variables is reducible.

Suppose $J$ has degree $2$.
It then satisfies either case $1$ or case $2$ of Proposition~1.1 from \cite{MM18}.
We can rule out case $2$, since if $a, b$ have degrees higher than $2$ then $J$ contains only reducible polynomials, and if $a, b$ have degree $1$ then $\lin{Q_{1}} = \lin{Q_{2}}$, and these are spanned by $a, b, x, y$.

Suppose $J$ has degree $3$.
It satisfies one of the eight cases from Theorem~2.1 of \cite{MM18}.
We can rule out case $1$ and case $2$, since quadratic polynomials in these ideals cannot be irreducible.
We can rule out case $3$ since every irreducible polynomial in such ideals have the same linear space.
For the remaining cases, we focus only on the generators of degree at most $2$.
Cases $4$ and $5$ are ruled out by irreducibility.
The other cases are ruled out since they have no generators of degree less than $3$.

Suppose $J$ has degree $4$, so it satisfies one of the $23$ cases of Theorem~3.1 in \cite{MM18}.
We can rule out cases $1, 2, 3, 4$ by irreducibility.
We can rule out case $5$ since our polynomials satisfy $\lin{Q_{1}} \neq \lin{Q_{2}}$.
The other cases have generators of higher degree, so we focus on the low degree parts.
In each of the cases, the degree two parts are either empty, or can only contain reducible polynomials.
This rules out all the cases.

\section{Covering Lemmas}

In this section, we state and prove the combinatorial covering lemmas that we will need in the Sylvester-Gallai proof in \cref{sec:sg-robust}.

We begin with an easy proposition on the basic covering of a set.
The following lemma states that given a family of large subsets, we can find a small subfamily such that every subset has large intersection with some element of the subfamily.

\begin{proposition} \label{proposition:cover-basic}
Let $\cU$ be a set of size at most $n$, and let $0 < \nu < 1$ be a constant.
Let $S_1, \ldots, S_t \subset U$ be a family of subsets of $\cU$ such that for each $i \in [t]$, we have $\abs{S_i} \geq \nu n$.
There exists a subset $\cH \subset [t]$ satisfying $|\cH| \leq 2/\nu$ such that for any $i \in [t]$, 
$$ \abs{S_i \cap \bigcup_{j \in \cH} S_j} \geq \nu n/2. $$
In particular, there exists $j \in \cH$ such that $|S_i \cap S_j| \geq \nu^2 n/4$.    
\end{proposition}

\begin{proof}
    We construct the subset $\cH$ greedily.
    Start with $\cH := \emptyset$.
    As long as there is an index $i \in [t] \setminus \cH$ such that the intersection of $S_{i}$ and $\cup_{j \in \cH} S_{j}$ is less than $\nu n/ 2$, we add $i$ to $\cH$.
    Each time we add an index to $\cH$, we increase the size of $\cup_{i \in \cH} S_i $ by at least $\nu n/ 2$, thus we have at any point of the process above we have
    $$ |\cH| \cdot \nu n/2 \leq \abs{\bigcup_{i \in \cH} S_i} \leq n \then \abs{\cH} \leq 2 / \nu.$$
    Therefore this iterative process must end after at most $2/\nu$ steps, and at the end of the process we must have $\abs{\cH} \leq 2/\nu$.
    Further, at the end of the process, for every $i \not\in \cH$, we must have that
    $$ \abs{S_i \cap \bigcup_{j \in \cH} S_j} \geq \nu n/2. $$
    By the pigeonhole principle, there is an index $j \in \cH$ such that $\abs{S_i \cap S_j} \geq \dfrac{\nu n}{2\abs{\cH}} \geq  \nu^2 n/4$.
\end{proof}

We now apply the above lemma to obtain a larger subfamily of sets that form a ``double cover'' of the elements of the universe, that is, such that every subset in our family has large simultaneous intersection with two elements of the subfamily.

\begin{lemma}[Double covering lemma]\label{lemma:double-cover}
Let $\cU$ be a set of size at most $n$, and let $0 < \nu < 1$ be a constant.
Let $S_1, \ldots, S_t \subset U$ be a family of subsets of $\cU$ such that for each $i \in [t]$, we have $\abs{S_i} \geq \nu n$.
There exist $h \leq 2/\nu$ and subsets $\cH, \cH_1, \ldots, \cH_h \subset [t]$ such that the following properties hold:
\begin{itemize}
    \item $|\cH| = h$ and $\cH$ satisfies the properties of \cref{proposition:cover-basic}
    \item $\abs{\cH_i} \leq 8/\nu$
    \item $\cH \cap \bigcup_{j \in [h]} \cH_j = \emptyset$
    \item for any $k \in [t]$ such that $k \not\in \cH \cup \bigcup_{i \leq h} \cH_i$, there exist $i \in \cH$ and $j \in \cH_i$ such that 
    $$|S_i \cap S_j \cap S_k| \geq \nu^{3} n/4^{3}$$
\end{itemize}
\end{lemma}

\begin{proof}
    For every $i$, define $S'_{i}$ to be a subset of $S_{i}$ of size exactly $\nu n$.
    For any $i, j, k$ we have $\abs{S_{i} \cap S_{j} \cap S_{k}} \geq \abs{S'_{i} \cap S'_{j} \cap S'_{k}}$.
    If the result is proved for the family $S'_{1}, \dots, S'_{t}$, then the result also holds for the original family $S_{1}, \dots, S_{t}$, with the same $\cH$ and $\cH_{i}$ acting as witnesses.
    We can thus assume that the original family consists of sets of size exactly $\nu n$.

    In order to obtain $\cH$, we apply \cref{proposition:cover-basic} to the family.
    After relabelling the subsets $S_{i}$, we can assume that $\cH = \bc{1, 2, \dots, h}$.
    
    We now partition the remaining subsets $S_{h+1}, \dots, S_{t}$ as follows.
    For each $i > h$, let $g_{i}$ be an index between $1$ and $h$ such that $\abs{S_{i} \cap S_{g_{i}}} \geq \nu^{2} n / 4$.
    The construction of $\cH$ guarantees the existence of at least one such index for every $i$, and if there are multiple such indices, we pick one arbitrarily.
    Define $\cG_{k} := \setbuild{i}{g_{i} = k}$, so each $\cG_{k}$ is a collection of subsets that has large intersection with $S_{k}$.
    
    In order to construct $\cH_{1}, \dots, \cH_{h}$, we apply \cref{proposition:cover-basic} to each of the families $\cG_{1}, \dots, \cG_{h}$.
    We will show how to construct $\cH_1$, as the construction of the other $\cH_i$'s are analogous.
    
    In this case, the set $\cU$ is $S_1$, which is of size $\abs{S_{1}} = \nu n$.
    The sets that make up the family are $S_{j} \cap S_{1}$ for every $j$ in $\cG_{1}$, which have size at least $\nu^{2} n / 4$.
    The parameter is $\nu / 4$.
    Applying \cref{proposition:cover-basic} we obtain $\cH_{1}$, a subset of $\cG_{1}$ of size at most $8 / \nu$.
    For any $j \in \cG_{1}$, either $j \in \cH_{1}$ or there is a $k \in \cH_{1}$ such that $\abs{S_{j} \cap S_{k} \cap S_{1}} \geq \nu^{3} n / 4^{3}$.
    We similarly get $\cH_{2}, \dots, \cH_{h}$.
    
    The first, second, and third claimed properties follow by construction.
    The fourth property follows by the argument in the previous paragraph, and the fact that the $\cG_{i} \setminus \cH_{i}$ cover the indices $\bc{h+1, \dots, n}$ that are not in any $\cH_{i}$.
\end{proof}

\end{document}